%% file: main.tex
\documentclass[envcountsame]{article}
\usepackage{ucs}
\usepackage[utf8x]{inputenc}
\usepackage{amsfonts,amsmath,amssymb,yhmath, amsbsy, amsthm}
\usepackage{mathrsfs}
\usepackage{color}
\usepackage[all]{xy}
\usepackage{stmaryrd}
\usepackage{tikz}
\usepackage[english]{babel}
\usepackage{url}

\input{macros}

\title{The Biequivalence\\ of Locally Cartesian Closed Categories\\ and
Martin-L\"of Type Theories}
\author{Pierre Clairambault and Peter Dybjer\\
University of Bath and Chalmers University of
  Technology}
\date{}
\begin{document}
\maketitle
\begin{abstract} {Seely's paper \emph{Locally cartesian closed
categories and type theory} contains a well-known result in
categorical type theory: that the category of locally cartesian closed
categories is equivalent to the category of Martin-L\"of type theories
with $\Pi, \Sigma$, and extensional identity types. However, Seely's
proof relies on the problematic assumption that substitution in types
can be interpreted by pullbacks. Here we prove a corrected version of
Seely's theorem: that the B\'enabou-Hofmann interpretation of
Martin-L\"of type theory in locally cartesian closed categories yields
a biequivalence of 2-categories. To facilitate the technical
development we employ categories with families as a substitute for
syntactic Martin-L\"of type theories. As a second result we prove that
if we remove $\Pi$-types the resulting categories with families are
biequivalent to left exact categories.

(The present paper without the appendix appears in the Proceedings of 
Typed Lambda Calculus and Applications, Novi Sad, Serbia, 1-3 June 2011.)
  }
\end{abstract}
\input{introtlcafinal}

\input{cwftlcafinal}
\input{forgetfultlcafinal}

\input{hofmanntlcafinal}
\input{biequivalencetlcafinal}



\appendix
\input{appendixtlca}
\end{document}

%% file: macros.tex
\newtheorem{lemma}{\textsc{Lemma}}
\newtheorem{theorem}{\textsc{Theorem}}
\newtheorem{definition}{\textsc{Definition}}
\newtheorem{proposition}{\textsc{Proposition}}

\definecolor{Red}{cmyk}{0,1,1,0}
\definecolor{Gray}{cmyk}{0,0,0,.3}

\newcommand{\Cwf}{\ensuremath{\mathbf{CwF}}}
\newcommand{\Cat}{\ensuremath{\mathbf{Cat}}}
\newcommand{\Ind}{\ensuremath{\mathbf{Ind}}}

\newcommand{\FL}{\ensuremath{\mathbf{FL}}}
\newcommand{\Fam}{\ensuremath{\mathbf{Fam}}}
\newcommand{\ML}{\ensuremath{\mathbf{ML}}}
\newcommand{\LCC}{\ensuremath{\mathbf{LCC}}}

\newcommand{\CwF}{\Cwf}
\newcommand{\CWFLCC}{\CwF^{\Id_{\mathrm{ext}}\Sigma\Pi}_{\mathrm{dem}}}
\newcommand{\CWFFL}{\CwF^{\Id_{\mathrm{ext}}\Sigma}_{\mathrm{dem}}}

\newcommand{\C}{\mathbb{C}}
\newcommand{\D}{\mathbb{D}}
\newcommand{\E}{\mathbb{E}}

\newcommand{\dom}{\mathrm{dom}}
\newcommand{\cod}{\mathrm{cod}}
\newcommand{\cext}{\!\cdot\!}

\newcommand{\tto}{\Rightarrow}
\newcommand{\ext}[1]{\langle #1 \rangle}
\newcommand{\subst}[1]{\langle #1 \rangle}
\renewcommand{\bar}[1]{\overline{#1}}

\newcommand{\iso}{\cong}
\newcommand{\natto}{\stackrel{\bullet}{\to}}
\newcommand{\modto}{\stackrel{\bullet\bullet}{\to}}
\newcommand{\arrow}[1]{{#1}^\rightarrow}

\newcommand{\dep}{\mathbf{T}}
\newcommand{\hof}[1]{T^{#1}}
\newcommand{\hf}[1]{\overrightarrow{#1}}
\renewcommand{\hat}[1]{\widehat{#1}}
\newcommand{\applyopen}[2]{\{ #1 \}  #2 }
\newcommand{\substopen}[2]{ #1 ^*#2}

\newcommand{\adj}{\dashv}

\newcommand{\nilc}{[]}

\newcommand{\p}{\mathrm{p}}
\newcommand{\q}{\mathrm{q}}

\newcommand{\snocvs}[2]{\langle #1, #2 \rangle}
\newcommand{\Te}[2]{#1 \vdash #2}
\newcommand{\comp}{\circ}

\newcommand{\id}{\mathrm{id}}
\newcommand{\Ty}{\mathrm{Type}}
\newcommand{\Sub}[2]{#1[#2]}
\newcommand{\sub}[2]{#1[#2]}
\newcommand{\ap}{\mathrm{ap}}
\newcommand{\indexed}[1]{\boldsymbol{#1}}
\newcommand{\Id}{\mathrm{I}}
\newcommand{\refl}{\mathrm{r}}
\newcommand{\pair}{\mathrm{pair}}

\newcommand{\Inverse}{\mathrm{Inv}}

%% file: introtlcafinal.tex
\section{Introduction}

It is ``well-known" that locally cartesian closed categories (lcccs) are
equivalent to Martin-L\"of's intuitionistic type theory \cite{martinlof:hannover,martinlof:padova}. But how
{\em known} is it really? Seely's original proof \cite{seely:lccc} contains
a flaw, and the papers by Curien \cite{CurienPL:subi} and Hofmann
\cite{hofmann:csl} who address this flaw only show that Martin-L\"of type
theory can be interpreted in locally cartesian closed categories, but
not that this interpretation is an equivalence of categories provided
the type theory has $\Pi, \Sigma,$ and extensional identity types. Here we
complete the work and fully rectify Seely's result except that we do
not prove an equivalence of categories but a \emph{biequivalence}
of 2-categories. In fact, a significant part of the endeavour has been
to find an appropriate formulation of the result, and in particular to
find a suitable notion analogous to Seely's ``interpretation of
Martin-L\"of theories".

\paragraph{Categories with families and democracy.} Seely turns a given Martin-L\"of theory into a category where the objects are {\em closed} types and the morphisms from type $A$ to type $B$ are functions of type $A \to B$. Such categories are the objects of Seely's ``category of Martin-L\"of theories". 

Instead of syntactic Martin-L\"of theories we shall employ {\em
  categories with families (cwfs)} \cite{dybjer:torino}. A cwf is a
pair $(\C,T)$ where $\C$ is the category of contexts and explicit
substitutions, and $T : \C^{op} \to \Fam$ is a functor, where
$T(\Gamma)$ represents the family of sets of terms indexed by types in context
$\Gamma$ and $T(\gamma)$ performs the substitution of $\gamma$ in
types and terms. Cwf is an appropriate substitute for syntax for
dependent types: its definition unfolds to a variable-free calculus of
explicit substitutions \cite{dybjer:torino}, which is like
Martin-L\"of's \cite{martinlof:gbg92,tasistro:lic} except that
variables are encoded by projections. 
One advantage of this approach compared to Seely's is that we get a natural definition of morphism of cwfs, which preserves the structure of cwfs up to isomorphism. In contrast Seely's notion of ``interpretation of Martin-L\"of  theories'' is defined indirectly via the construction of an lccc associated with a Martin-L\"of theory, and basically amounts to a functor preserving structure between the corresponding lcccs, rather than directly as something which preserves all the ``structure" of Martin-L\"of theories.

To prove our biequivalences we require that our cwfs are {\em democratic}. This means that each context is {\em represented} by a type. Our results require us to build local cartesian closed structure in the category of contexts. To this end we use available constructions on types and terms, and by democracy such constructions can be moved back and forth between types and contexts. Since Seely works with closed types only he has no need for democracy.

\paragraph{The coherence problem.}
Seely interprets type substitution in Martin-L\"of theories as pullbacks in lcccs. However, this is problematic, since type substitution is already defined by induction on the structure of types, and thus fixed by the interpretation of the other constructs of type theory. It is not clear that the pullbacks can be chosen to coincide with this interpretation. 




In the paper \emph{Substitution up to isomorphism}  \cite{CurienPL:subi} Curien
describes the fundamental nature of this problem. He sets out
\begin{quotation}
  ... to solve a difficulty arising from a
  mismatch between syntax and semantics: in locally cartesian closed
  categories, substitution is modelled by pullbacks (more generally
  pseudo-functors), that is, only up to isomorphism, unless split
  fibrational hypotheses are imposed.  ... but not all semantics do
  satisfy them, and in particular not the general description of the
  interpretation in an arbitrary locally cartesian closed category. In
  the general case, we have to show that the isomorphisms between
  types arising from substitution are \emph{coherent} in a sense familiar to category theorists.
\end{quotation}
To solve the problem Curien introduces a calculus with explicit substitutions for Martin-L\"of type theory, with special terms witnessing applications of the type equality rule. In this calculus type equality can be interpreted as isomorphism in lcccs. The remaining coherence problem is to show that Curien's calculus is equivalent to the usual formulation of Martin-L\"of type theory, and Curien proves this result by cut-elimination.

Somewhat later, Hofmann \cite{hofmann:csl} gave an alternative
solution based on a technique which had been used by B\'enabou
\cite{BenabouJ:fibcfn} for constructing a {\em split} fibration from
an arbitrary fibration. In this way Hofmann constructed a model of
Martin-L\"of type theory with $\Pi$-types, $\Sigma$-types, and
(extensional) identity types from a locally cartesian closed
category. Hofmann used categories with attributes (cwa) in the sense
of Cartmell \cite{cartmell:apal} as his notion of model. In fact, cwas
and cwfs are closely related: the notion of cwf arises by
reformulating the axioms of cwas to make the connection with the usual
syntax of dependent type theory more transparent. Both cwas and cwfs
are split notions of model of Martin-L\"of type theory, hence the
relevance of B\'enabou's construction.

However, Seely wanted to prove an equivalence of categories. Hofmann
conjectured \cite{hofmann:csl}:
\begin{quotation}
We have now constructed a cwa over $\cal{C}$ which can be shown to be equivalent to $\cal{C}$ in some suitable 2-categorical sense.
\end{quotation}
Here we spell out and prove this result, and thus fully rectify Seely's theorem.  It should be apparent from what follows that this is not a trivial exercise. In our setting the result is a biequivalence analogous to B\'enabou's (much simpler) result: that the 2-category of fibrations (with non-strict morphisms) is biequivalent to the 2-category of split fibrations (with non-strict morphisms).

While carrying out the proof we noticed that if we remove $\Pi$-types the resulting 2-category of cwfs is biequivalent to the 2-category of left exact (or finitely complete) categories. We present this result in parallel with the main result.

\paragraph{Plan of the paper.} 
An equivalence of categories consists of a pair of functors which are
inverses up to natural isomorphism. Biequivalence is the appropriate
notion of equivalence for bicategories \cite{leinster:bicategories}. Instead of functors we have
{\em pseudofunctors} which only preserve identity and composition up
to isomorphism. Instead of natural isomorphisms we have {\em pseudonatural transformations} which are inverses up to {\em invertible
  modification}.

A 2-category is a strict bicategory, and the remainder of the paper
consists of constructing two biequivalences of 2-categories. In
Section 2 we introduce cwfs and show how to turn a cwf into an indexed
category. In Section 3 we define the 2-categories $\CWFFL$ of
democratic cwfs which support extensional identity types and
$\Sigma$-types and $\CWFLCC$ which also support $\Pi$-types. We
also define the notions of pseudo cwf-morphism and pseudo
cwf-transformation.  In Section 4 we define the 2-categories $\FL$ of
left exact categories and $\LCC$ of locally cartesian closed
categories. We show that there are forgetful 2-functors $U :
\CWFFL \to \FL$ and $U : \CWFLCC \to \LCC$. In section 5 we construct
the pseudofunctors $H : \FL \to \CWFFL$ and $H : \LCC \to \CWFLCC$
based on the B\'enabou-Hofmann construction. In section 6 we prove
that $H$ and $U$ give rise to the biequivalences of $\FL$ and $\CWFFL$
and of $\LCC$ and $\CWFLCC$.

An appendix containing the full proof of the biequivalences can be found at
\url{http://www.cse.chalmers.se/~peterd/papers/categorytypetheory.html/}.

\paragraph{Acknowledgement.} We are grateful to the anonymous
reviewers for several useful remarks which have helped us improve the
paper. We would also like to acknowledge the support of the
(UK) EPSRC grant RC-CM1025 for the first author and of the (Swedish) Vetenskapsr\aa det grant ``Types for Proofs and Programs'' for the second author.

%% file: cwftlcafinal.tex
\section{Categories with Families}

\subsection{Definition}

\begin{definition}
Let $\Fam$ be the category of families of sets defined as follows. 
An object is a pair $(A,B)$ where $A$ is a set and $B(x)$ is a family of sets indexed by $x \in A$. A morphism with source $(A,B)$
and target $(A',B')$
is a pair consisting of a function $f:A\to A'$ and a family of functions $g(x):B(x) \to B'(f(x))$ indexed by $x\in A$.
\end{definition}

Note that $\Fam$ is equivalent to the arrow category $\mathbf{Set}^{\to}$.

\begin{definition}A \textbf{category with families (cwf)} consists of the following data:
\begin{itemize}
\item A base category $\C$. Its objects represent \emph{contexts} and
  its morphisms represent \emph{substitutions}. The identity map is
  denoted by $\id : \Gamma \to \Gamma$ and the composition of maps
  $\gamma : \Delta \to \Gamma$ and $\delta : \Xi \to \Delta : \Xi \to
  \Gamma$ is denoted by $\gamma \circ \delta$ or more briefly by
  $\gamma \delta : \Xi \to \Gamma$.
\item A functor $T:\C^{op}\to \Fam$. $T(\Gamma)$ is a pair, where the
  first component represents the set $\Ty(\Gamma)$ of {\em types} in context
  $\Gamma$, and the second component represents the type-indexed
  family $(\Gamma\vdash A)_{A\in \Ty(\Gamma)}$ of sets of {\em terms} in
  context $\Gamma$. We write $a : \Gamma \vdash A$ for a term $a \in
  \Gamma \vdash A$. Moreover, if $\gamma$ is a morphism in $\C$, then
  $T(\gamma)$ is a pair consisting of the {\em type substitution} function
  $A \mapsto A[\gamma]$ and the type-indexed family of {\em term
  substitution} functions $a \mapsto a[\gamma]$.
\item A \emph{terminal object} $\nilc$ of $\C$ which represents the
  \emph{ empty context} and a terminal map $\subst{} : \Delta \to
  \nilc$ which represents the \emph{empty substitution}.
\item A \emph{context comprehension} which to an object $\Gamma$ in
  $\C$ and a type $A \in \Ty(\Gamma)$ associates an
  object $\Gamma\cext A$ of $\C$, a morphism $\p_A:\Gamma\cext A\to
  \Gamma$ of $\C$ and a term $\q\in \Gamma\cext A\vdash A[\p]$ such the
  following universal property holds: for each object $\Delta$ in
  $\C$, morphism $\gamma : \Delta \to \Gamma$, and term $a\in \Delta
  \vdash A[\gamma]$, there is a unique morphism $\theta =
  \ext{\gamma,a} : \Delta \to \Gamma\cext A$, such that $\p_A\circ\theta
  = \gamma$ and $\q[\theta] = a$. (We remark that a related notion of comprehension for hyperdoctrines was introduced by Lawvere \cite{lawvere:hyperdoctrines}.)
\end{itemize}
\end{definition}

The definition of cwf can be presented as a system of axioms and inference rules for a variable-free generalized algebraic formulation of the most basic rules of dependent type theory \cite{dybjer:torino}. The correspondence with standard syntax is explained by Hofmann \cite{hofmann:cambridge} and the equivalence is proved in detail by Mimram \cite{mimram:cwf}. 
The easiest way to understand this correspondence might be as a translation between the standard lambda calculus based syntax of dependent type theory and the language of cwf-combinators. In one direction the key idea is to translate a variable (de Bruijn number) to a projection of the form $\q[\p^n ]$. In the converse direction, recall that the cwf-combinators yield a calculus of explicit substitutions whereas substitution is a meta-operation in usual lambda calculus. When we translate cwf-combinators to lambda terms, we execute the explicit substitutions, using the equations for substitution in types and terms as rewrite rules.
The equivalence proof is similar to the proof of the equivalence of cartesian closed categories and the simply typed lambda calculus.

We shall now define what it means that a cwf supports extra structure corresponding to the rules for the various type formers of Martin-L\"of type theory.

\begin{definition}
A cwf {\em supports (extensional) identity types} provided the following conditions hold:
\begin{description}
\item[Form.] If $A \in \Ty(\Gamma)$ and $a, a' : \Gamma\vdash A$, there is $\Id_A(a, a') \in \Ty(\Gamma)$;
\item[Intro.] If $a:\Gamma\vdash A$, there is $\refl_{A,a}:\Gamma\vdash \Id_A(a, a)$;
\item[Elim.] If $c:\Gamma\vdash \Id_{A}(a,a')$ then $a=a'$ and $c= \refl_{A, a}$.
\end{description}
Moreover, we have stability under substitution: if
$\delta : \Delta \to \Gamma$ then
\begin{eqnarray*}
\Id_A(a,a')[\delta] &=& \Id_{A[\delta]}(a[\delta], a'[\delta])\\
\refl_{A, a}[\delta] &=& \refl_{A[\delta], a[\delta]}
\end{eqnarray*}
\end{definition}

\begin{definition}
A cwf {\em supports $\Sigma$-types} iff the following conditions hold:
\begin{description}
\item[Form.] If $A\in \Ty(\Gamma)$ and $B\in \Ty(\Gamma \cext A)$, there is $\Sigma(A,B)\in \Ty(\Gamma)$,
\item[Intro.] If $a:\Gamma\vdash A$ and $b:\Gamma\vdash B[\ext{\id,a}]$, there is $\pair(a,b):\Gamma\vdash \Sigma(A,B)$,
\item[Elim.] If $a:\Gamma \vdash \Sigma(A,B)$, there are $\pi_1(a):\Gamma\vdash A$ and $\pi_2(a):\Gamma\vdash B[\ext{\id,\pi_1(a)}]$ such that
\begin{eqnarray*}
\pi_1(\pair(a,b)) &=& a\\
\pi_2(\pair(a,b)) &=& b\\
\pair(\pi_1(c),\pi_2(c)) &=& c
\end{eqnarray*}
\end{description}
Moreover, we have stability under substitution:
\begin{eqnarray*}
\Sigma(A,B)[\delta] &=& \Sigma(A[\delta],B[\ext{\delta\circ \p,\q}])\\
\pair(a,b)[\delta] &=& \pair(a[\delta],b[\delta])\\
\pi_1(c)[\delta] &=& \pi_1(c[\delta])\\
\pi_2(c)[\delta] &=& \pi_2(c[\delta])
\end{eqnarray*}
\end{definition}
Note that in a cwf which supports extensional identity types and $\Sigma$-types surjective pairing, $\pair(\pi_1(c),\pi_2(c)) = c$, follows from the other conditions \cite{martinlof:padova}.

\begin{definition}
A cwf {\em supports $\Pi$-types} iff the following conditions hold:
\begin{description}
\item[Form.] If $A\in \Ty(\Gamma)$ and $B\in \Ty(\Gamma \cext A)$, there is $\Pi(A,B)\in \Ty(\Gamma)$.
\item[Intro.] If $b: \Te{{\Gamma}\cext{A}}{B}$, there is $\lambda(b) : \Te{\Gamma}{\Pi (A,B)}$.
\item[Elim.] If $c : \Te{\Gamma}{\Pi(A,B)}$ and $a:\Te{\Gamma}{A}$ then there is a term $\ap(c,a) :
\Te{\Gamma}{\Sub{B}{\snocvs{\id}{a}}}$ such that
\begin{eqnarray*}
\ap(\lambda (b),a) &=& \sub{b}{\snocvs{\id}{a}} : \Te{\Gamma}{\Sub{B}{\snocvs{\id}{a}}}\\
c &=& \lambda (\ap(\sub{c}{\p},\q )) :  \Te{\Gamma}{\Pi
(A,B)}
\end{eqnarray*}
\end{description}
Moreover, we have stability under substitution:
\begin{eqnarray*}
\Sub{\Pi (A,B)}{\gamma} &=&
\Pi (\Sub{A}{\gamma},\Sub{B}{\snocvs{\gamma \comp \p}{\q}})\\
\sub{\lambda (b)}{\gamma} &=& \lambda (\sub{b}{\snocvs{\gamma \comp \p}{\q}})\\
\sub{\ap(c,a)}{\gamma} &=& \ap(\sub{c}{\gamma},\sub{a}{\gamma})
\end{eqnarray*}
\end{definition}

\begin{definition}
A cwf $(\C, T)$ is \emph{democratic} iff for each object $\Gamma$ of $\C$ there is $\bar{\Gamma} \in \Ty(\nilc)$ and an isomorphism
 $\Gamma \iso_{\gamma_\Gamma} \nilc \cext \bar{\Gamma}$. Each substitution $\delta : \Delta \to \Gamma$ can then be represented
by the term $\bar{\delta} = q[\gamma_\Gamma\delta\gamma_\Delta^{-1}] : \nilc \cext \bar{\Delta}\vdash \bar{\Gamma}[p]$.
\end{definition}
Democracy does not correspond to a rule of Martin-L\"of type theory. However, a cwf generated inductively by the standard rules of Martin-L\"of type theory with a one element type  $\mathrm{N_1}$ and $\Sigma$-types is democratic, since we can associate $\mathrm{N_1}$ to the empty context and the closed type $\Sigma x_1 : A_1. \cdots .\Sigma x_n : A_n$ to a context $x_1 : A_1, \ldots , x_n : A_n$ by  induction on $n$.

\subsection{The Indexed Category of Types in Context}

We shall now define the indexed category associated with a cwf. This will play a crucial role and in particular introduce the notion of \emph{isomorphism} of types.

\begin{proposition}[The Context-Indexed Category of Types]
If $(\C, T)$ is a cwf, then we  can define a functor $\indexed{T}: \C^{op} \to \Cat$ as follows:
\begin{itemize}
\item The objects of $\indexed{T}(\Gamma)$ are types in $\Ty(\Gamma)$. If $A,B \in \Ty(\Gamma)$, then a morphism in $\indexed{T}(\Gamma)(A,B)$ is a morphism 
$\delta: \Gamma\cext A\to \Gamma\cext B$ in $\C$ such that $\p\delta = \p$.
\item If $\gamma : \Delta \to \Gamma$ in $\C$, then $\indexed{T}(\gamma) : \Ty(\Gamma) \rightarrow \Ty(\Delta)$ maps an object $A \in \Ty(\Gamma)$ to $A[\gamma]$ and a
morphism $\delta  : \Gamma \cext A \to \Gamma \cext B$ to $\ext{\p, \q[\delta\ext{\gamma \circ \p,\q}]} : \Delta \cext A[\gamma] \to \Delta\cext B[\gamma]$.
\end{itemize}
\end{proposition}
We write $A \iso_\theta B$ if $\theta: A\to B$ is an isomorphism in $\indexed{T}(\Gamma)$.
If $a : \Gamma\vdash A$, we write $\{\theta\}(a) = q[\theta\subst{id, a}] : \Gamma\vdash B$ for the \emph{coercion} of $a$ to type $B$ and $a =_\theta b$ if $a = \{\theta\}(b)$. Moreover, we get an alternative formulation of democracy.
\begin{proposition} $(\C, T)$ is democratic iff the functor from
  $\indexed{T}(\nilc)$ to $\C$, which maps a closed type $A$ to the
  context $\nilc \cext A$, is an equivalence of categories. 
\end{proposition}
Seely's category $\ML$ of Martin-L\"of theories \cite{seely:lccc} is
essentially the category of categories $\indexed{T}(\nilc)$ of closed
types.

\paragraph{Fibres, slices and lcccs.} Seely's interpretation of type
theory in lcccs relies on the idea that a type $A\in \Ty(\Gamma)$ can
be interpreted as its \emph{display map}, that is, a morphism with
codomain $\Gamma$. For instance, the type $\texttt{list}(n)$ of lists
of length $n:\mathtt{nat}$ would be mapped to the function $l:
\mathtt{list} \to \mathtt{nat}$ which to each list associates its
length. Hence, types and terms in context $\Gamma$ are
interpreted in the \emph{slice category} $\C/\Gamma$, since terms are interpreted as global sections. Syntactic types
are connected with types-as-display-maps by the following result, an analogue of which was one of the cornerstones of Seely's paper.

\begin{proposition}
  If $(\C, T)$ is democratic and supports extensional identity and $\Sigma$-types, then $\indexed{T}(\Gamma)$ and $\C/\Gamma$ are equivalent
  categories for all $\Gamma$.
\label{fibres_vs_slices}
\end{proposition}
\begin{proof}
To each object (type) $A$ in $\indexed{T}(\Gamma)$ we associate the
object $\p_{A}$ in $\C/\Gamma$. A morphism from $A$ to $B$ in
$\indexed{T}(\Gamma)$ is by definition a morphism from $\p_{A}$ to
$\p_{B}$ in $\C/\Gamma$.

Conversely, to each object (morphism) $\delta: \Delta \to
\Gamma$ of $\C/\Gamma$ we associate a type in $\Ty(\Gamma)$. This is the inverse image $x : \Gamma \vdash \Inverse(\delta)(x)$ which is defined type-theoretically by
$$
\Inverse(\delta)(x) = \Sigma y : \bar{\Delta}. \Id_{\bar{\Gamma}}(\bar{x},\bar{\delta}(y))
$$
written in ordinary notation. In cwf combinator notation it becomes
$$
\Inverse(\delta) =
\Sigma(\bar{\Delta}[\ext{}],
\Id_{\bar{\Gamma}[\ext{}]}(\q[\gamma_\Gamma\p],\bar{\delta}[\ext{\ext{},\q}]) \in
\Ty(\Gamma)
$$
These associations yield an equivalence of categories since $\p_{\Inverse(\delta)}$ and $\delta$ are isomorphic in $\C/\Gamma$.
\end{proof}

It is easy to see that $\indexed{T}(\Gamma)$ has binary products if the cwf supports $\Sigma$-types and exponentials if it supports $\Pi$-types. Simply define $A \times B = \Sigma(A,B[\p])$ and $B^A = \Pi(A,B[\p])$. Hence by Proposition 9 it follows that $\C/\Gamma$ has products and $\C$ has finite limits in any democratic cwf which supports extensional identity types and $\Sigma$-types. If it supports $\Pi$-types too, then $\C/\Gamma$ is cartesian closed and $\C$ is locally cartesian closed.


\section{The 2-Category of Categories with Families}

\subsection{Pseudo Cwf-Morphisms}

A notion of {\em strict cwf-morphism} between cwfs $(\C, T)$ and $(\C', T')$ was defined by Dybjer \cite{dybjer:torino}. It is a pair $(F, \sigma )$, where 
$F: \C \to \C'$ is a functor and 
$\sigma : T \natto T'F$ is a natural transformation of family-valued functors,
such that terminal objects and context comprehension are preserved on the nose.
Here we need a weak version where the terminal object, context comprehension, and substitution of types and terms of a cwf are only preserved up to isomorphism. The pseudo-natural transformations needed to prove our biequivalences will be families of cwf-morphisms which do not preserve cwf-structure on the nose.

The definition of pseudo cwf-morphism will be analogous to that of
\emph{strict} cwf-morphism, but cwf-structure will only be preserved
up to coherent isomorphism.

\begin{definition}A \textbf{pseudo cwf-morphism} from $(\C, T)$ to $(\C', T')$ is a pair $(F, \sigma)$ where:
\begin{itemize}
\item $F: \C \to \C'$ is a functor,
\item For each context $\Gamma$ in $\C$, $\sigma_\Gamma$ is a $\Fam$-morphism from $T \Gamma$ to $T' F \Gamma$. We will write $\sigma_\Gamma(A) : \Ty'(F\Gamma)$ for the type component and $\sigma^A_\Gamma(a) : F\Gamma \vdash \sigma_\Gamma(A)$ for the term component of this morphism.
\end{itemize}
The following preservation properties must be satisfied:
\begin{itemize}
\item Substitution is preserved: For each context $\delta : \Delta \to \Gamma$ in $\C$ and $A\in \Ty(\Gamma)$, there is an isomorphism of types
$\theta_{A, \delta} : \sigma_\Gamma(A)[F\delta] \to \sigma_\Delta(A[\delta])$ such that substitution on terms is also preserved, that is,
$\sigma_\Delta^{A[\gamma]}(a[\gamma]) =_{\theta_{A,\gamma}} \sigma_\Gamma^A(a) [F \gamma]$.
\item The terminal object is preserved: $F\nilc$ is terminal.
\item Context comprehension is preserved: $F(\Gamma\cext A)$ with the projections $F(\p_A)$ and $\{\theta_{A, \p}^{-1}\}(\sigma_{\Gamma\cext A}^{A[p]}(\q_A))$
is a context comprehension of $F\Gamma$ and $\sigma_{\Gamma}(A)$. Note that the universal property on context comprehensions provides a unique
isomorphism $\rho_{\Gamma, A}: F(\Gamma\cext A) \to F\Gamma\cext \sigma_\Gamma(A)$ which preserves projections.
\end{itemize}
These data must satisfy naturality and coherence laws which amount to the fact that if we extend $\sigma_\Gamma$ to a functor
$\indexed{\sigma}_\Gamma : \indexed{T}(\Gamma) \to \indexed{T'}F(\Gamma)$, then $\indexed{\sigma}$ is a pseudo natural transformation from $\indexed{T}$ to $\indexed{T'}F$.  This functor is defined by
$\indexed{\sigma}_\Gamma(A) = \sigma_\Gamma(A)$ on an object $A$ and $\indexed{\sigma}_\Gamma(f) = \rho_{\Gamma, B} F(f) \rho^{-1}_{\Gamma, A}$ on a morphism $f: A\to B$.
\end{definition}

A consequence of this definition is that all cwf structure is preserved.

\begin{proposition}
Let $(F, \sigma)$ be a pseudo cwf-morphism from $(\C, T)$ to $(\C', T')$.
\begin{itemize}
\item[(1)] Then substitution extension is preserved: for all $\delta : \Delta \to \Gamma$ in $\C$ and $a:\Delta \vdash A[\delta]$, we have $F(\subst{\delta,a}) = \rho_{\Gamma,A}^{-1}\subst{F\delta,\applyopen{\theta_{A,\delta}^{-1}}{(\sigma_\Delta^{A[\delta]}(a))}}$.
\item[(2)] Redundancy terms/sections: for all $a\in \Gamma\vdash A$, $\sigma_\Gamma^A(a) = \q[\rho_{\Gamma, A} F(\ext{\id, a})]$.
\end{itemize}
\label{properties}
\end{proposition}

If $(F, \sigma) : (\C_0, T_0) \to (\C_1, T_1)$ and $(G, \tau) : (\C_1, T_1) \to (\C_2, T_2)$ are two pseudo cwf-morphisms, we define their composition
$(G, \tau)(F, \sigma)$ as $(GF, \tau\sigma)$ where:
\begin{eqnarray*}
(\tau\sigma)_\Gamma(A) &=& \tau_{F\Gamma}(\sigma_\Gamma(A))\\
(\tau\sigma)_\Gamma^A(a) &=& \tau_{F\Gamma}^{\sigma_\Gamma(A)}(\sigma_\Gamma^A(a))
\end{eqnarray*}
The families $\theta^{GF}$ and $\rho^{GF}$ are obtained from $\theta^F, \theta^G$ and $\rho^F$ and $\rho^G$ in the obvious way.
The fact that these data satisfy the necessary coherence and naturality conditions basically amounts to the stability of pseudonatural transformation under composition.
There is of course an identity pseudo cwf-morphism whose components are all identities, which is obviously neutral for composition. So, there is a category of cwfs and pseudo cwf-morphisms.

Since the isomorphism $(\Gamma \cext A) \cext B \iso \Gamma \cext \Sigma(A, B)$ holds in an arbitrary cwf which supports $\Sigma$-types, it follows that pseudo cwf-morphisms automatically preserve $\Sigma$-types, since they preserve context comprehension.
However, if cwfs support other structure, we need to define what it means that cwf-morphisms preserve this extra structure up to isomorphism.

\begin{definition}
Let $(F, \sigma)$ be a pseudo cwf-morphism between cwfs $(\C,T)$ and $(\C',T')$ which support identity types, $\Pi$-types, and democracy, respectively.
\begin{itemize}
\item $(F, \sigma)$ {\em preserves identity types} provided $\sigma_\Gamma(\Id_A(a,a')) \cong \Id_{\sigma_\Gamma(A)}(\sigma_\Gamma^A(a),\sigma_\Gamma^A(a))$;
\item $(F, \sigma)$ {\em preserves $\Pi$-types} provided $\sigma_\Gamma(\Pi(A,B)) \cong \Pi(\sigma_\Gamma(A),\sigma_{\Gamma \cext A}(B)[\rho_{\Gamma,A}^{-1}])$;
\item $(F, \sigma)$ {\em preserves democracy} provided $\sigma_{\nilc}(\bar{\Gamma}) \cong_{d_\Gamma} \bar{F\Gamma}[\subst{}]$,
and the following diagram commutes:
\[
\xymatrix@R=10pt{
F\Gamma	\ar[rr]^{F\gamma_\Gamma}
	\ar[d]_{\gamma_{F\gamma}}&&
F(\nilc\cext \bar{\Gamma})
	\ar[d]^{\rho_{\nilc, \bar{\Gamma}}}\\
\nilc\cext\bar{F\Gamma}
	\ar@{<-}[r]^{\subst{\subst{}, \q}}&
F\nilc\cext \bar{F\Gamma}[\subst{}]
	\ar@{<-}[r]^{d_\Gamma}&
F\nilc\cext \sigma_{\nilc}(\bar{\Gamma})
}
\]
\end{itemize}
\end{definition}
These preservation properties are all stable under composition and thus yield several different 2-categories of structure-preserving pseudo cwf-morphisms.

\subsection{Pseudo Cwf-Transformations} 

\begin{definition}[Pseudo cwf-transformation]
Let $(F, \sigma)$ and $(G, \tau)$ be two cwf-morphisms from $(\C, T)$ to $(\C', T')$. A \emph{pseudo cwf-transformation} from  $(F, \sigma)$ to 
$(G, \tau)$ is a pair $(\phi, \psi)$ where $\phi: F \natto G$ is a natural transformation, and for each $\Gamma$ in $\C$ and $A\in \Ty(\Gamma)$,
a morphism $\psi_{\Gamma, A} : \sigma_\Gamma(A) \to \tau_\Gamma(A)[\phi_\Gamma]$ in $\indexed{T'}(F\Gamma)$, natural in $A$ and 
such that the following diagram commutes:
\[
\xymatrix@C=60pt@R=20pt{
\sigma_\Gamma(A)[F\delta]
	\ar[r]^{\indexed{T'}(F\delta)(\psi_{\Gamma, A})}
	\ar[d]^{\theta_{A, \delta}}&
\tau_\Gamma(A)[\phi_\Gamma F(\delta)]
	\ar[d]^{\indexed{T'}(\phi_\Delta)(\theta'_{A, \delta})}\\
\sigma_\Delta(A[\delta])
	\ar[r]_{\psi_{\Delta, A[\delta]}}&
\tau_\Delta(A[\delta])[\phi_\Delta]
}
\]
where $\theta$ and $\theta'$ are the isomorphisms witnessing preservation of substitution in types in the definition of pseudo cwf-morphism.
\end{definition}

Pseudo cwf-transformations can be composed both vertically (denoted by $(\phi', \psi')(\phi, \psi)$) and
horizontally (denoted by $(\phi', \psi')\star(\phi, \psi)$), and these compositions
are associative and satisfy the interchange law.  Note that just
as coherence and naturality laws for pseudo cwf-morphisms ensure that
they give rise to pseudonatural transformations (hence morphisms of
indexed categories) $\indexed{\sigma}$ to $\indexed{\tau}$, this definition
exactly amounts to the fact that pseudo cwf-transformations between $(F, \sigma)$ and $(F, \tau)$ correspond to
modifications from $\indexed{\sigma}$ to $\indexed{\tau}$.

\subsection{2-Categories of Cwfs with Extra Structure}

\begin{definition} Let $\CWFFL$ be the 2-category of small democratic categories with families which support
  extensional identity types and $\Sigma$-types. The 1-cells are
  cwf-morphisms preserving democracy and extensional identity types (and
  $\Sigma$-types automatically) and the 2-cells are pseudo
  cwf-transformations.

  Moreover, let $\CWFLCC$ be the sub-2-category of $\CWFFL$ where also
  $\Pi$-types are supported and preserved.
\end{definition}

%% file: forgetfultlcafinal.tex
\section{Forgetting Types and Terms}
\begin{definition} Let $\FL$ be the 2-category of small categories with finite
  limits (left exact categories). The $1$-cells are functors preserving finite
  limits (up to isomorphism) and the $2$-cells are natural transformations.

Let $\LCC$ be the 2-category of small locally cartesian closed categories. The $1$-cells are functors preserving local cartesian closed structure (up to isomorphism), and the $2$-cells are natural transformations.
\end{definition}

$\FL$ is a sub(2-)category of the 2-category of categories: we do not provide a choice of finite limits. Similarly, $\LCC$ is a sub(2-)category of $\FL$.
The first component of our biequivalences will be \emph{forgetful}
$2$-functors.

\begin{proposition}\label{forgetful-functors}
The forgetful $2$-functors 
\begin{eqnarray*}
U &:& \CWFFL \to \FL\\
U &:& \CWFLCC \to \LCC
\end{eqnarray*}
defined as follows on 0-, 1-, and 2-cells
\begin{eqnarray*}
U(\C, T) &=& \C\\
U(F, \sigma) &=& F\\
U(\phi, \psi) &=& \phi
\end{eqnarray*}
are well-defined.
\end{proposition}
\begin{proof}
By definition $U$ is a $2$-functor from $\Cwf$ to $\Cat$, it remains to prove that it sends a cwf in $\CWFFL$ to $\FL$ and
a cwf in $\CWFLCC$ to $\LCC$, along with the corresponding properties for $1$-cells and $2$-cells.

For 0-cells we already proved as corollaries of Proposition \ref{fibres_vs_slices} that if $(\C, T)$ supports $\Sigma$-types, identity types and democracy, 
then $\C$ has finite limits; and if $(\C, T)$ also supports $\Pi$-types, then $\C$ is an lccc. 

For 1-cells we need to prove that if $(F, \sigma)$ preserves identity types and democracy, then $F$ preserves finite limits; 
and if $(F, \sigma)$ also preserves $\Pi$-types then $F$ preserves local exponentiation. 
Since finite limits and local exponentiation in $\C$ and $\C'$ can be defined by the inverse image construction, these two statements boil down
to the fact that if $(F, \sigma)$ preserves identity types and democracy then inverse images are preserved. Indeed we have an isomorphism 
$F(\Gamma\cext \Inverse(\delta)) \iso F\Gamma \cext \Inverse(F \delta)$. This
can be proved by long but mostly direct
calculations involving all components and coherence laws of pseudo cwf-morphisms. 



There is nothing to prove for $2$-cells.

\end{proof}

%% file: hofmanntlcafinal.tex
\section{Rebuilding Types and Terms}

Now, we turn to the reverse construction. We use the Bénabou-Hofmann construction to build a cwf from any finitely complete category,
then generalize this operation to functors and natural transformations, and show that this gives rise to a pseudofunctor.

\begin{proposition}
  There are pseudofunctors 
\begin{eqnarray*}
H &:&  \FL \to \CWFFL \\
H &:& \LCC \to \CWFLCC
\end{eqnarray*}
defined by
  \begin{eqnarray*}
    H \C &=& (\C, T_\C)\\
    H F &=& (F, \sigma_F)\\
    H \phi &=& (\phi, \psi_\phi)
  \end{eqnarray*}
on 0-cells, 1-cells, and 2-cells, respectively, and where $T_\C, \sigma_F,$ and $\psi_\phi$ are defined in the following three subsections.
\end{proposition}

\begin{proof}
The remainder of this Section contains the proof. We will in turn show the action on 0-cells, 1-cells, 2-cells, and then prove pseudofunctoriality of $H$.
\end{proof}

\subsection{Action on 0-Cells}

As explained before, it is usual (going back to Cartmell \cite{cartmell:apal}) to represent a type-in-context $A \in \Ty(\Gamma)$ in a category as a 
{\em display map} \cite{taylor:pfm}, that is, as an object $\p_{A}$ in $\C/\Gamma$. A term $\Te{\Gamma}{A}$ is then represented as a section of the 
display map for $A$, that is, a morphism $a$ such that $\p_{A} \circ a = \id_\Gamma$.
Substitution in types is then represented by pullback. This is essentially the technique used by Seely for interpreting Martin-L\"of
type theory in lcccs. However, as we already mentioned, it leads to a coherence problem.

To solve this problem Hofmann \cite{hofmann:csl} used a construction due to B\'enabou \cite{BenabouJ:fibcfn}, which from any fibration builds an equivalent \emph{split}
fibration. Hofmann used it to build a category with attributes (cwa) \cite{cartmell:apal} from a locally cartesian closed category. He then showed that this cwa supports $\Pi,\Sigma$, and extensional identity types. This technique essentially
amounts to associating to a type $A$, not only a display map, but a whole family of display maps, one for each substitution instance $A[\delta]$. In other words, we choose a pullback
square for every possible substitution and this choice is split, hence solving the coherence problem.
As we shall explain below this family takes the form of a functor, and we refer to it as a {\em functorial family}.

Here we reformulate Hofmann's construction using cwfs. See Dybjer \cite{dybjer:torino} for the correspondence between cwfs and cwas.

\begin{lemma}
  Let $\C$ be a category with terminal object. Then we can build a
  democratic cwf $(\C,T_\C)$ which supports $\Sigma$-types. If $\C$
  has finite limits, then $(\C,T_\C)$ also supports extensional
  identity types. If $\C$ is locally cartesian closed, then $(\C, T_\C)$
  also supports $\Pi$-types.
\end{lemma}

\begin{proof}
We only show the definition of types and terms in $T_\C(\Gamma)$. This construction is essentially the same as Hofmann's \cite{hofmann:csl}.

A \emph{type} in $\Ty_\C(\Gamma)$ is a \emph{functorial family}, that is, a functor 
$\hf{A}:\C/\Gamma\to \arrow{\C}$ 
such that $\cod\circ \hf{A} = \dom$ and if
\vspace{-15pt} 
$\xymatrix@R=5pt@C=0pt{
\Omega\ar[dr]_{\delta\alpha}\ar[rr]^{\alpha}&&\Delta\ar[dl]^{\delta}\\
&\Gamma}$
is a morphism in $\mathbb{C}/\Gamma$, then $\hf{A}(\alpha)$ is a pullback square:
\[
\xymatrix{
\ar[r]^{\hf{A}(\delta,\alpha)}
                             \ar[d]_{\hf{A}(\delta \alpha)}	&
\ar[d]^{\hf{A}(\delta)}\\
\Omega\ar[r]_\alpha & \Delta
}
\]
Following Hofmann, we denote the upper arrow of the square  by $\hf{A}(\delta,\alpha)$.


A {\em term} $a:\Gamma\vdash \hf{A}$ is a section of $\hf{A}(\id_\Gamma)$, that is, a morphism $a:\Gamma\to \Gamma\cext \hf{A}$ such that 
$\hf{A}(\id_\Gamma) a = \id_\Gamma$, where we have defined context extension by
$\Gamma\cext \hf{A} = \dom(\hf{A}(\id_\Gamma))$. Interpreting types as functorial families makes it easy to define
substitution in types. Substitution in terms is obtained by exploiting the universal property of pullback squares,
yielding a functor $T_\C : \C^{op} \to \Fam$.
\end{proof}

Note that $(\C, T_\C)$ is a \emph{democratic} cwf since to any context $\Gamma$ we can associate a functorial family $\hat{\subst{}}: \C/\nilc \to \arrow{C}$, where
$\subst{}: \Gamma \to \nilc$ is the terminal projection. The isomorphism $\gamma_\Gamma : \Gamma \to \nilc\cext \hat{\subst{}}$ is just $\id_\Gamma$.

\subsection{Action on 1-Cells}

Suppose that $\C$ and $\C'$ have finite limits and that $F:\C \to \C'$ preserves them. As described in the previous section, $\C$ and $\C'$ give
rise to cwfs $(\C, T_\C)$ and $(\C', T_{\C'})$. In order to extend $F$ to a pseudo cwf-morphism, we need to define, for each object $\Gamma$ in $\C$,
a $\Fam$-morphism $(\sigma_F)_\Gamma : T_\C(\Gamma) \to T_{\C'} F(\Gamma)$. Unfortunately, unless $F$ is full, it does not seem possible
to embed faithfully a functorial family $\hf{A}: \C/\Gamma \to \arrow{\C}$ into a functorial family over $F\Gamma$ in $\C'$. However, there is such an
embedding for display maps (just apply $F$) from which we will freely regenerate a functorial family from the obtained display map.

\paragraph{The ``hat" construction.}
As remarked by Hofmann, any morphism $f : \Delta \to \Gamma$ in a category $\C$ with a (not necessarily split) choice of finite limits generates a functorial family $\hat{f}:\C/\Gamma\to
\C^\rightarrow $. If $\delta : \Delta \to \Gamma$ then $\hat{f}(\delta) = \delta^*(f)$, where $\delta^*(f)$ is obtained by taking the pullback of $f$
along $\delta$ ($\delta^*$ is known as the \emph{pullback functor}):
\[
\xymatrix{
~       \ar[d]_{\delta^*(f)}
        \ar[r]&
~       \ar[d]^{f}\\
\Delta  \ar[r]_\delta&
\Gamma
}
\]
Note that we can always choose pullbacks such that $\hat{f}(\id_\Gamma) = \id_\Gamma^*(f) = f$. If 
$\xymatrix@R=5pt@C=5pt{
\Omega\ar[dr]_{\delta \alpha}\ar[rr]^{\alpha}&&\Delta\ar[dl]^{\delta}\\
&\Gamma
}$is a morphism in $\C/\Gamma$, we define $\hat{f}(\alpha)$ as the left square in the following diagram:
\[
\xymatrix{
~       \ar[r]^{\hat{f}(\delta,\alpha)}
        \ar[d]_{\hat{f}(\delta\alpha)}&
~       \ar[d]^{\hat{f}(\delta)}
        \ar[r]&
~       \ar[d]^f\\
\Delta' \ar[r]_\alpha&
\Delta  \ar[r]_\delta&
\Gamma
}
\]
This is a pullback, since both the outer square and the right square
are pullbacks.

\paragraph{Translation of types.}
The hat construction can be used to extend $F$ to types:
\[
\sigma_F(\hf{A}) = \hat{F(\hf{A}(\id))}
\]
Note that $F(\Gamma\cext \hf{A}) = F(\dom(\hf{A}(\id))) = \dom(F(\hf{A}(\id))) = \dom(\sigma_\Gamma(\hf{A})(\id)) = F\Gamma \cext \sigma_\Gamma(\hf{A})$, so context comprehension is preserved on the
nose. However, substitution on types is \emph{not} preserved on the nose. Hence we have to define a coherent family of isomorphisms $\theta_{\hf{A}, \delta}$. 

\paragraph{Completion of cwf-morphisms.} Fortunately, whenever $F$ preserves finite limits there is a canonical way to generate all the remaining data. 

\begin{lemma}[Generation of isomorphisms]
  Let $(\C, T)$ and $(\C', T')$ be two cwfs, $F: \C \to \C'$ a functor
  preserving finite limits, $\sigma_\Gamma : \Ty(\Gamma) \to
  \Ty'(F\Gamma)$ a family of functions, and $\rho_{\Gamma, A} :
  F(\Gamma\cext A) \to F\Gamma \cext \sigma_\Gamma(A)$ a family of
  isomorphisms such that $\p \rho_{\Gamma, A} = F \p$. Then there
  exists an unique choice of functions $\sigma_\Gamma^A$ on terms and
  of isomorphisms $\theta_{A, \delta}$ such that $(F, \sigma)$ is a
  pseudo cwf-morphism.
\label{completion_cwfmorphisms}
\end{lemma}
\begin{proof}
By item \emph{(2)} of Proposition \ref{properties}, the unique way to extend $\sigma$ to terms is to
set $\sigma_\Gamma^A(a) = \q[\rho_{\Gamma, A} F(\ext{\id, a})]$. To generate $\theta$, we use the two squares below:
\[\hspace{-5pt}
\xymatrix@R=20pt@C=35pt{
F\Delta \cext \sigma_\Gamma({A})[F\delta]       \ar[r]^{\ext{(F \delta) \p\q}}
        \ar[d]_{\p}&
F\Gamma\cext \sigma_\Gamma({A}) \ar[d]^{\p}\\
F\Delta \ar[r]_{F\delta}&
F\Gamma
}
\xymatrix@R=20pt@C=80pt{
F\Delta\cext \sigma_\Delta({A}[\delta]) \ar[r]^{\rho_{\Gamma, A} F(\ext{\delta \p, \q}) \rho_{\Delta, A[\delta]}^{-1}}
        \ar[d]_{\p}&
F\Gamma\cext \sigma_\Gamma(A)   \ar[d]^{\p}\\
F\Delta \ar[r]_{F\delta}&
F\Gamma
}
\]
The first square is a substitution pullback. The second is a pullback because $F$ preserves finite limits and $\rho_{\Gamma, A}$ and $\rho_{\Delta, A[\delta]}$ are isomorphisms.
The isomorphism $\theta_{A, \delta}$ is defined as the unique mediating morphism from the first to the second. It follows from the universal property of pullbacks that
the family $\theta$ satisfies the necessary naturality and coherence conditions.
There is no other choice for $\theta_{A, \delta}$, because if $(F, \sigma)$ is 
a pseudo cwf-morphism with families of isomorphisms $\theta$ and $\rho$, then 
$\rho_{\Gamma, A} F(\subst{\delta \p, \q}) \rho_{\Delta, A[\delta]}^{-1} \theta_{A, \delta} = \subst{(F\delta)\p, \q}$.
Hence if $F$ preserves finite limits, $\theta_{A, \delta}$ must coincide with the mediating morphism.
\end{proof}

\paragraph{Preservation of additional structure.} As a pseudo cwf-morphism, $(F, \sigma_F)$ automatically preserves $\Sigma$-types. Since the democratic structure of $(\C, T_\C)$ and $(\C', T_{\C'})$ is trivial
it is clear that it is preserved by $(F, \sigma_F)$. To prove that it also preserves type constructors, we use the following proposition.

\begin{proposition}
Let $(F, \sigma)$ be a pseudo cwf-morphism between $(\C, T)$ and $(\C', T')$ supporting $\Sigma$-types and democracy. Then:
\begin{itemize}
\item If $(\C, T)$ and $(\C', T')$ both support identity types, then $(F, \sigma)$ preserves identity types provided $F$ preserves finite limits.
\item If $(\C, T)$ and $(\C', T')$ both support $\Pi$-types, then $(F, \sigma)$ preserves $\Pi$-types provided $F$ preserves local exponentiation.
\end{itemize}
\label{preservation_structure}
\end{proposition}
\begin{proof}
For the first part it remains to prove that if $F$ preserves finite limits, then 
  $(F, \sigma)$ preserves identity types. Since  $a, a'\in \Gamma\vdash
  A$, $\p_{\Id_A(a, a')}:\Gamma\cext \Id_A(a, a') \to \Gamma$ is an
  equalizer of $\subst{\id, a}$ and $\subst{\id, a'}$ and $F$ preserves equalizers, it follows that $F(\p_{\Id_A(a, a')})$ is an equalizer of
  $\subst{\id, \sigma_\Gamma^A(a)}$ and $\subst{\id,
    \sigma_\Gamma^A(a')}$, and by uniqueness of equalizers it is
  isomorphic to $\Id_{\sigma_\Gamma(A)}(\sigma_\Gamma^A(a),
  \sigma_\Gamma^A(a'))$.  

The proof of preservation of $\Pi$-types
  exploits in a similar way the uniqueness (up to iso) of
  ``$\Pi$-objects" of $A\in \Ty(\Gamma)$ and $B\in \Ty(\Gamma\cext
  A)$.
\end{proof}

\subsection{Action on $2$-Cells}

Similarly to the case of $1$-cells, under some conditions a natural transformation $\phi: F\natto G$ where $(F, \sigma)$ and $(G, \tau)$ are
pseudo cwf-morphisms can be completed to a pseudo cwf-transformation $(\phi, \psi_\phi)$, as stated below.

\begin{lemma}[Completion of pseudo cwf-transformations]
Suppose $(F, \sigma)$ and $(G, \tau)$ are pseudo cwf-morphisms from  $(\C, T)$ to $(\C', T)$  such that $F$ and $G$ preserve finite limits and $\phi: F\natto G$
is a natural transformation, then there exists a family of morphisms $(\psi_\phi)_{\Gamma, A} : \sigma_\Gamma(A) \to \tau_\Gamma(A)[\phi_\Gamma]$
such that $(\phi, \psi_\phi)$ is a pseudo cwf-transformation from $(F, \sigma)$ to $(G, \tau)$.

\label{completion_transformations}
\end{lemma}
\begin{proof}
We set $\psi_{\Gamma, A} = \subst{\p, \q[\rho'_{\Gamma, A}\phi_{\Gamma\cext A}\rho_{\Gamma, A}^{-1}]}: F\Gamma\cext \sigma_\Gamma A\to F\Gamma\cext \tau_\Gamma(A)[\phi_\Gamma]$.
To check the coherence law, we apply the universal property of a well-chosen pullback square (exploiting the fact that $G$ preserves finite limits).
\end{proof}

This completion operation on $2$-cells commutes with units and both notions of composition, as will be crucial to prove pseudofunctoriality of $H$:

\begin{lemma}
If $\phi: F \natto G$ and $\phi' : G\natto H$, then
\begin{eqnarray*}
(\phi', \psi_{\phi'})(\phi, \psi_\phi) &=& (\phi'\phi, \psi_{\phi'\phi})\\
(\phi, \psi_\phi)\star 1 &=& (\phi\star 1, \psi_{\phi\star 1})\\
1\star (\phi, \psi_\phi) &=& (1\star \phi, \psi_{1\star \phi})\\
(\phi', \psi_{\phi'})\star(\phi, \psi_\phi)&=& (\phi'\star\phi, \psi_{\phi'\star\phi})
\end{eqnarray*}
\label{comm_completion}
whenever these expressions typecheck.
\end{lemma}
\begin{proof}
Direct calculations.
\end{proof}

\subsection{Pseudofunctoriality of $H$}

Note that $H$ is \emph{not} a functor, because for any $F: \C\to \D$ with finite limits and functorial family $\hf{A}$ over $\Gamma$ (in $\C$), 
$\sigma_\Gamma(\hf{A})$ forgets all information on $\hf{A}$ except its display map $\hf{A}(\id)$, and later extends $F(\hf{A}(\id))$ to 
an independent functorial family.
However if $F: \C\to \D$ and $G: \D \to \E$ preserve finite limits, the two pseudo cwf-morphisms $(G, \sigma^G)\circ (F, \sigma^F) = (GF, \sigma^G \sigma^F)$ and $(GF, \sigma^{GF})$
are related by the pseudo cwf-transformation $(1_{GF}, \psi_{1_{GF}})$, which is obviously an isomorphism. The coherence laws only involve vertical 
and horizontal compositions of units and pseudo cwf-transformations obtained by completion, hence they are easy consequences of Lemma \ref{comm_completion}.

%% file: biequivalencetlcafinal.tex
\section{The Biequivalences}

\begin{theorem}
We have the following biequivalences of $2$-categories.

\[
\begin{array}{lcr}
\xymatrix{
\FL\ 
\ar@<0.5ex>[r]^{H\ \ \ \ }&
\ \CWFFL
\ar@<0.5ex>[l]^{U\ \ \ \ }
}&~~~~~~~~&
\xymatrix{
\LCC\ 
\ar@<0.5ex>[r]^{H\ \ \ \ }&
\ \CWFLCC
\ar@<0.5ex>[l]^{U\ \ \ \ }
}
\end{array}
\]
\end{theorem}
\begin{proof}
Since $UH = \mathrm{Id}$ (the identity 2-functor) it suffices to construct pseudonatural transformations of pseudofunctors:
\[
\xymatrix{
\mathrm{Id}
\ar@<0.5ex>[r]^{\eta}&
HU
\ar@<0.5ex>[l]^{\epsilon}
}
\]
which are inverse up to invertible modifications. Since $HU(\C, T) = (\C, \hof{\C})$, these pseudonatural transformations are families of equivalences of cwfs:
\[
\xymatrix{
(\C, T)
\ar@<0.5ex>[r]^{\eta_{(\C, T)}}&
(\C, \hof{\C})
\ar@<0.5ex>[l]^{\epsilon_{(\C, T)}}
}
\]
which satisfy the required conditions for pseudonatural transformations.


\paragraph{Construction of $\eta_{(\C, T)}$.} Using Lemma \ref{completion_cwfmorphisms}, we just need to define a base functor, which will be $\mathrm{Id}_\C$,
and a family $\sigma^\eta_\Gamma$ which translates types (in the sense of $T$) to functorial families. This is easy, since types in the cwf $(\C, T)$
come equipped with a chosen behaviour under substitution.
Given $A\in \Ty(\Gamma)$, we define:
\begin{eqnarray*}
\sigma^\eta_\Gamma(A)(\delta) &=& \p_{A[\delta]}\\
\sigma^\eta_\Gamma(A)(\delta, \gamma) &=& \ext{\gamma \p, \q}
\end{eqnarray*}
For each pseudo cwf-morphism $(F, \sigma)$, the pseudonaturality square relates two pseudo cwf-morphisms whose base functor
is $F$. Hence, the necessary invertible pseudo cwf-transformation is obtained using Lemma \ref{completion_transformations} from
the identity natural transformation on $F$. The coherence conditions are straightforward consequences of Lemma \ref{comm_completion}.

\paragraph{Construction of $\epsilon_{(\C, T)}$.} As for $\eta$, the base functor for $\epsilon_{(\C, T)}$ is $\mathrm{Id}_\C$. 
Using Lemma \ref{completion_cwfmorphisms} again we need, for each context $\Gamma$, a function $\sigma^\epsilon_\Gamma$
which given a functorial family $\hf{A}$ over $\Gamma$ will build a syntactic type $\sigma^\epsilon_\Gamma(\hf{A}) \in \Ty(\Gamma)$.
In other terms, we need to find a syntactic representative of an arbitrary display map, that is, an arbitrary morphism in $\C$. We use the inverse image:
\[
\sigma^\epsilon_\Gamma(\hf{A}) = \Inverse(\hf{A}(\id)) \in \Ty(\Gamma)
\]
The family $\epsilon$ is pseudonatural for the same reason as $\eta$ above.

\paragraph{Invertible modifications.} For each cwf $(\C, T)$, we need to define invertible pseudo cwf-transformations $m_{(\C, T)} : (\epsilon\eta)_{(\C, T)} \to \id_{(\C, T)}$ and
$m'_{(\C, T)} : (\eta\epsilon)_{(\C, T)} \to \id_{(\C, T)}$. As pseudo cwf-transformations between pseudo cwf-morphisms with the same base functor, their first component will be the identity natural
transformation, and the second will be generated by Lemma \ref{completion_transformations}. The coherence law for modifications is a consequence
of Lemma \ref{comm_completion}.
\end{proof}

%% file: appendixtlca.tex
\section{Proofs of Section 2}

\begin{lemma}[Composition of coercions]
If $(\C, T)$ is a cwf, $\Gamma$ a context in $\C$ and $\theta_1: A\to B$ and $\theta_2 : B \to C$ are isomorphisms in $\indexed{T}(\Gamma)$, then
for all $a: \Gamma \vdash A$,
\[
\{\theta_2\}(\{\theta_1\}(a)) = \{\theta_2\theta_1\}(a)
\]
\label{composition_coercions_appendix}
\end{lemma}
\begin{proof}
Direct calculation, using the definition of coercions and manipulation of cwf combinators.
\begin{eqnarray*}
\{\theta_2\}(\{\theta_1\}(a)) 	&=& \q[\theta_2\subst{\id, \q[\theta_1\subst{\id, a}]}]\\
				&=& \q[\theta_2\subst{\p \theta_1 \subst{\id, a}, \q[\theta_2\subst{\id, a}]}]\\
				&=& \q[\theta_2\subst{\p, \q} \theta_1 \subst{\id, a}]\\
				&=& \q[\theta_2\theta_1 \subst{\id, a}]\\
				&=& \{\theta_2\theta_1\}(a)
\end{eqnarray*}
\end{proof}

\begin{lemma}
Let $(\C, T)$ be a democratic cwf with $\Sigma$-types and identity types, then for each $\delta : \Delta\to \Gamma$ there is an isomorphism
$\alpha_{\delta}$ in $\C/\Gamma$:
\[
\xymatrix@R=20pt{
\Gamma\cext \Inverse(\delta)
	\ar@/^/[rr]^{\alpha_\delta}
	\ar@/_/@{<-}[rr]_{\alpha_\delta^{-1}}
	\ar[dr]_{\p}&&
\Delta	\ar[dl]^{\delta}\\
&\Gamma
}
\]
\label{inverse_prop_appendix}
\end{lemma}
\begin{proof}
Recall that $\Inverse(\delta) = \Sigma(\bar{\Delta}[\ext{}], \Id_{\bar{\Gamma}[\ext{}]}(\q[\gamma_\Gamma\p],\bar{\delta}[\ext{\ext{},\q}]) \in \Ty(\Gamma)$.
We define:
\begin{eqnarray*}
\alpha_\delta &=& \gamma_\Delta^{-1}\subst{\subst{}, \pi_1(\q)}\\
\alpha_\delta^{-1} &=&  \subst{\delta, \pair(\q[\gamma_\Delta], \refl_{\bar{\Gamma}[\subst{}]})}
\end{eqnarray*}
A straightforward calculation proves that this typechecks, and that $\alpha_\delta \alpha_\delta^{-1} = \id_\Delta$. For the other equality,
we have $\alpha_\delta^{-1} \alpha_\delta = \subst{\delta\gamma_\Delta^{-1}\subst{\subst{}, \pi_1(\q)}, \pair(\pi_1(\q), \refl_{\bar{\Gamma}[\subst{}]})}$.
But by property of extensional identity types $\q[\gamma_\Gamma\p]$ and $\bar{\delta}[\ext{\ext{},\pi_1(\q)}]$ are equal terms in context $\Gamma\cext \Inverse(\delta)$, so
$\gamma_\Gamma^{-1} \subst{\subst{}, \q[\gamma_\Gamma \p]} = \p$ and $\gamma_\Gamma^{-1} \subst{\subst{}, \bar{\delta}[\ext{\ext{},\pi_1(\q)}]} = \delta\gamma_\Delta^{-1}\subst{\subst{}, \pi_1(\q)}$ are equal substitutions.
Likewise, $\refl_{\bar{\Gamma}[\subst{}]} = \pi_2(\q)$ by uniqueness of identity proofs, therefore $\alpha_\delta^{-1} \alpha_\delta = \id_{\Gamma\cext \Inverse(\delta)}$.
\end{proof}

\begin{proposition}
Let $(\C, T)$ be a democratic cwf with $\Sigma$-types and identity types, then for all context $\Gamma$ the categories $\indexed{T}(\Gamma)$ and
$\C/\Gamma$ are equivalent.
\end{proposition}
\begin{proof}
For each context $\Gamma$, the functor $F_\Gamma : \dep(\Gamma) \to \C/\Gamma$ defined by $F_\Gamma(A) = \p_A$ on objects and $F_\Gamma(f) = f$ on morphisms
in clearly full and faithful, but it is also essentially surjective. Indeed, for any object $\delta : \Delta \to \Gamma$
in $\C/\Gamma$, we have by Lemma \ref{inverse_prop_appendix} an isomorphic $F_{\Gamma}(\Inverse(\delta)) = \p_{\Inverse(\delta)}$, hence $F_{\Gamma}$ 
is an equivalence of categories.
\end{proof}

\section{Proofs of Section 3}

\subsection{Properties of pseudo cwf-morphisms}

Let us first mention that in the definition of a pseudo cwf-morphism $(F, \sigma)$ from $(\C, T)$ to $(\C', T')$, the fact that $\indexed{\sigma}$
is a pseudonatural transformation from $\indexed{T}$ to $\indexed{T'}$ amounts to the satisfaction of the following coherence and naturality laws.
\begin{itemize}
\item \emph{Identity.} For all $A\in \Ty(\Gamma)$, we have $\theta_{A,\id} = \id_{F\Gamma \cext \sigma_\Gamma(A)}$,
\item \emph{Coherence.} For all $\delta: \Xi \to \Delta$ and $\gamma : \Delta \to \Gamma$, the following diagram commutes.
\[
\xymatrix
{
F \Xi \cext \sigma_\Gamma(A)[F(\gamma \delta)]
\ar[dr]_{\indexed{T'}(F\delta)(\theta_{A,\gamma})}
\ar[rr]^{\theta_{A,\gamma \delta}}
&&
F \Xi \cext \sigma_\Xi(A[\gamma\delta])\\
&F \Xi \cext \sigma_\Delta(A[\gamma])[F(\delta)]
\ar[ur]_{\theta_{A[\gamma],\delta}}
}
\]
\item \emph{Naturality.} For all $\delta: \Delta \to \Gamma$ in $\C$, $A, B\in \Ty(\Gamma)$ and $f: A \to B$ in $\indexed{T}(\Gamma)$, the following
diagram commutes in $\indexed{T'}(F\Delta)$.
\[
\xymatrix@C=60pt{
\sigma_\Gamma(A)[F\delta]
        \ar[r]^{\theta_{A, \delta}}
        \ar[d]^{\indexed{T'}(F\delta)(\indexed{\sigma}_\Gamma(f))}&
\sigma_\Delta(A[\delta])
        \ar[d]^{\indexed{\sigma}_\Delta(\indexed{T}(\delta)(f))}\\
\sigma_\Gamma(B)[F\delta]
        \ar[r]^{\theta_{B, \delta}}&
\sigma_\Delta(B[\delta])
}
\]
\end{itemize}
This can be checked by simply unfolding the definition of a pseudonatural transformation.

\begin{proposition}
Any pseudo cwf-morphism $(F, \sigma)$ from $(\C, T)$ to $(\C', T')$ preserves substitution extension, in the following sense:
For all $\delta : \Delta \to \Gamma$ in $\C$ and $a:\Delta \vdash A[\delta]$, we have 
\[
F(\subst{\delta,a}) = \rho_{\Gamma,A}^{-1}\subst{F\delta,\applyopen{\theta_{A,\delta}^{-1}}{(\sigma_\Delta^{A[\delta]}(a))}}
\]
\label{substitution_extension_appendix}
\end{proposition}
\begin{proof}
First note that for each context $\Gamma$ in $\C$ and type $A\in \Ty(\Gamma)$, the isomorphism $\rho_{\Gamma, A}$ is defined as the unique morphism
preserving projections between the two context comprehensions of $F\Gamma$ and $\sigma_\Gamma A$, in other terms
$\rho_{\Gamma, A} = \subst{F(\p_A), \{\theta_{A, \p}^{-1}\}(\sigma_{\Gamma\cext A}^{A[\p]}(\q_A))}$, which implies that the projections are preserved
in the following sense.
\begin{eqnarray*}
F(\p_A) &=& \p_{\sigma_\Gamma A} \rho_{\Gamma, A}\\
\sigma_{\Gamma\cext A}^{A[\p]}(\q_A) &=& \{\theta_{A, \p}\}(\q_{\sigma_\Gamma A}[\rho_{\Gamma, A}])
\end{eqnarray*}
We now use it to prove the announced property.
Clearly, the required equality boils down to the following two equations.
\begin{eqnarray*}
\p \rho_{\Gamma, A} F(\subst{\delta,a}) &=& F\delta\\
\q[\rho_{\Gamma, A} F(\subst{\delta,a})]&=& \{\theta_{A, \delta}^{-1}\}(\sigma_\Delta^{A[\delta]}(a))
\end{eqnarray*}
The proof of the first equality is completely straightforward:
\begin{eqnarray*}
\p \rho_{\Gamma, A} F(\subst{\delta,a}) &=& F(\p) F(\subst{\delta,a})\\
                                        &=& F\delta
\end{eqnarray*}
However, the proof of the second is far more subtle and relies on many  properties of pseudo cwf-morphisms and cwf combinators:
\begin{eqnarray*}
\q[\rho_{\Gamma, A} F(\subst{\delta,a})]&=_1&\{\theta_{A, \p}^{-1}\}(\sigma_{\Gamma\cext A}^{A[\p]}(\q))[F(\subst{\delta,a})]\\
                                        &=_2&\q[\theta_{A, \p}^{-1}\subst{\id, \sigma_{\Gamma\cext A}^{A[\p]}(\q)}][F(\subst{\delta,a})]\\
                                        &=&\q[\theta_{A, \p}^{-1}\subst{F(\subst{\delta, a}), \sigma_{\Gamma\cext A}^{A[\p]}(\q)[F(\subst{\delta, a})]}]\\
                                        &=_3&\q[\theta_{A, \p}^{-1}\subst{F(\subst{\delta, a}),\applyopen{\theta_{A[\p], \subst{\delta, a}}^{-1}}{(\sigma_\Delta^{A[\delta]}(\q[\subst{\delta, a}]))}}]\\
                                        &=&\q[\theta_{A, \p}^{-1}\subst{F(\subst{\delta, a}),\applyopen{\theta_{A[\p], \subst{\delta, a}}^{-1}}{(\sigma_\Delta^{A[\delta]}(a))}}]\\
                                        &=&\q[\subst{\p, \q[\theta_{A, \p}^{-1}\subst{F(\subst{\delta, a})\p, \q}]}\subst{\id, \{\theta_{A[\p], \subst{\delta, a}}^{-1}\}(\sigma_\Delta^{A[\delta]}(a))}]\\
                                        &=_4&\q[\indexed{T'}(F(\subst{\delta, a}))(\theta_{A, \p}^{-1})\subst{\id, \{\theta_{A[\p], \subst{\delta, a}}^{-1}\}(\sigma_\Delta^{A[\delta]}(a))}]\\
                                        &=_2&\{\indexed{T'}(F(\subst{\delta, a}))(\theta_{A, \p}^{-1})\}(\{\theta_{A[\p], \subst{\delta, a}}^{-1}\}(\sigma_\Delta^{A[\delta]}(a)))\\
                                        &=_5&\{\theta_{A, \delta}^{-1}\}(\sigma_\Delta^{A[\delta]}(a))
\end{eqnarray*}
Equality (1) is by preservation of $\q$, equalities (2) by definition of coercions, equality (3) by preservation
of substitution on terms, equality (4) by definition of $\indexed{T'}$, equality (5) by the coherence requirement on $\theta$ and Lemma \ref{composition_coercions_appendix}. All
the other steps are by simple manipulations on cwf combinators.
\end{proof}

\begin{lemma}
If $(F, \sigma)$ is a pseudo cwf-morphism from $(\C, T)$ to $(\C', T')$, then its action on terms and sections is redundant: 
for all $a\in \Gamma\vdash A$, 
\[
\sigma_\Gamma^A(a) = \q[\rho_{\Gamma, A} F(\subst{\id, a})]
\]
\label{redundant_appendix}
\end{lemma}
\begin{proof}
This is a direct consequence of preservation of substitution extension, as follows:
\[
F(\subst{\id, a}) = \rho_{\Gamma, A}^{-1} \subst{\id, \applyopen{\theta_{A, \id}^{-1}}{\sigma_\Gamma^A(a)}}
\]
but $\theta_{A, \id} = \id$ by coherence of $\theta$, hence the result is proved.
\end{proof}

\begin{lemma}
If $(F, \sigma)$ is a pseudo cwf-morphism from $(\C, T)$ to $(\C', T')$ and $\theta: A\to B$ is a morphism in $\indexed{T}(\Gamma)$, then 
the coercion $\{\theta\}$ commutes with $\sigma$ in the following way, for each $a\in \Gamma\vdash A$:
\[
\sigma_\Gamma^B(\{\theta\}(a)) = \{\indexed{\sigma}_\Gamma(\theta)\}(\sigma_\Gamma^A(a))
\]
\label{comm_sigma_theta_appendix}
\end{lemma}
\begin{proof}
Direct calculation.
\begin{eqnarray*}
\sigma_\Gamma^B(\{\theta\}(a)) 	&=_1& \q[\rho_{\Gamma, B} F(\subst{\id, \{\theta\}(a)})]\\
				&=_2& \q[\rho_{\Gamma, B} F(\subst{\id, \q[\theta\subst{\id, a}]})]\\
				&=_3& \q[\rho_{\Gamma, B} F( \theta \subst{\id, a})]\\
				&=_4& \q[\indexed{\sigma}_\Gamma(\theta) \rho_{\Gamma, A} F(\subst{\id, a})]\\
				&=_2& \{\indexed{\sigma}_\Gamma(\theta)\}(\q[\rho_{\Gamma, A} F(\subst{\id, a})])\\
				&=_1& \{\indexed{\sigma}_\Gamma(\theta)\}(\sigma_\Gamma^A(a))
\end{eqnarray*}
Where (1) is by Lemma \ref{redundant_appendix}, (2) by definition of coercions, (3) by basic manipulation of cwf combinators and (4)
by definition of $\indexed{\sigma}$.
\end{proof}

\subsection{Composition of pseudo cwf-morphisms}

\begin{proposition}
Pseudo cwf-morphisms are stable under composition.
\end{proposition}
\begin{proof}
Let us first give a bit more details about how pseudo cwf-morphisms are composed.
If $(F, \sigma) : (\C_0, T_0) \to (\C_1, T_1)$ and $(G, \tau) : (\C_1, T_1) \to (\C_2, T_2)$ are two pseudo cwf-morphisms, we define their composition
$(G, \tau)(F, \sigma)$ as $(GF, \tau\sigma)$ where:
\begin{eqnarray*}
(\tau\sigma)_\Gamma(A) &=& \tau_{F\Gamma}(\sigma_\Gamma(A))\\
(\tau\sigma)_\Gamma^A(a) &=& \tau_{F\Gamma}^{\sigma_\Gamma(A)}(\sigma_\Gamma^A(a))
\end{eqnarray*}
If the other components of $(F, \sigma)$ are denoted by $\theta^F, \rho^F$ and those of $(G, \tau)$ by $\theta^G, \rho^G$, we define:
\[
\theta_{A, \delta} = \indexed{\tau}_{F\Delta}(\theta^F_{A, \delta})\theta^G_{\sigma_\Gamma(A), F\delta}\\
\]
All the components of $(G, \tau)(F, \sigma)$ are now defined, but we still have a number of conditions to prove.
\begin{itemize}
\item \emph{Preservation of substitution on terms.} Direct calculation, if $a: \Gamma\vdash A$ and $\delta : \Delta \to \Gamma$ in $\C_0$.
\begin{eqnarray*}
(\tau\sigma)_\Delta^{A[\delta]}(a[\delta]) 	&=_1& \tau_{F\Delta}^{\sigma_\Delta(A[\delta])}(\sigma_\Delta^{A[\delta]}(a[\delta]))\\
						&=_2& \tau_{F\Delta}^{\sigma_\Delta(A[\delta])}(\{\theta^F_{A, \delta}\}(\sigma_\Gamma^A(a)[F\delta]))\\
						&=_3& \q[\rho^G_{F\Delta, \sigma_\Delta(A[\delta])} G(\subst{\id, \{\theta^F_{A, \delta}\}(\sigma_\Gamma^A(a)[F\delta])})]\\
						&=_4& \q[\rho^G_{F\Delta, \sigma_\Delta(A[\delta])} G(\subst{\id, \q[\theta^F_{A, \delta}\subst{\id, \sigma_\Gamma^A(a)[F\delta]}]})\\
						&=_5& \q[\rho^G_{F\Delta, \sigma_\Delta(A[\delta])} G(\theta^F_{A, \delta} \subst{\id,\sigma_\Gamma^A(a)[F\delta]})]\\
						&=_6& \q[\indexed{\tau}_{F\Delta}(\theta^F_{A, \delta}) \rho_{F\Delta, \sigma_\Gamma(A)[F\delta]}^G G(\subst{\id,\sigma_\Gamma^A(a)[F\delta]})]\\
						&=_7& \q[\indexed{\tau}_{F\Delta}(\theta^F_{A, \delta}) \subst{\id, \q[\rho_{F\Delta, \sigma_\Gamma(A)[F\delta]}^G G(\subst{\id,\sigma_\Gamma^A(a)[F\delta]})]}]\\
						&=_4& \{\indexed{\tau}_{F\Delta}(\theta^F_{A, \delta})\}(\q[\rho_{F\Delta, \sigma_\Gamma(A)[F\delta]}^G G(\subst{\id,\sigma_\Gamma^A(a)[F\delta]})])\\
						&=_3& \{\indexed{\tau}_{F\Delta}(\theta^F_{A, \delta})\}(\tau_{F\Delta}^{\sigma_\Gamma(A)[F\delta]}(\sigma_\Gamma^A(a)[F\delta]))\\
						&=_2& \{\indexed{\tau}_{F\Delta}(\theta^F_{A, \delta})\}(\{\theta_{\sigma_\Gamma(A), F\delta}^G\}(\tau_{F\Gamma}^{\sigma_\Gamma(A)}(\sigma_\Gamma^A(a))[GF\delta]))\\
						&=_8& \{\theta_{A, \delta}\}(\tau_{F\Gamma}^{\sigma_\Gamma(A)}(\sigma_\Gamma^A(a))[GF\delta])\\
						&=_1& \{\theta_{A, \delta}\}((\tau\sigma)_\Gamma^A(a)[GF\delta])
\end{eqnarray*}
Equalities annotated by (1) come from the definition of $\tau\sigma$, (2) is preservation of substitution for $\sigma$ or $\tau$, (3) is Lemma \ref{redundant_appendix},
(4) is by definition of coercions, (5) uses $\p \theta^F_{A, \delta} = \p$ and basic manipulations with cwf combinators, (6) is by definition of $\indexed{\tau}$, 
(7) uses preservation of $\p$ by $(G, \tau)$ and basic manipulations with cwf combinators, and (8) is by definition of $\theta$.
\item \emph{Preservation of the terminal object.} Trivial from the preservation of the terminal object by $F$ and $G$.
\item \emph{Preservation of context comprehension.} Using preservation of context comprehension from $(F, \sigma)$ and $(G, \tau)$ we define:
\[
\xymatrix@C=40pt{
GF(\Gamma\cext A)
	\ar[r]^{G(\rho^F_{\Gamma, A})}&
G(F\Gamma\cext \sigma_\Gamma A)
	\ar[r]^{\rho^G_{F\Gamma, \sigma_\Gamma A}}&
GF\Gamma\cext (\tau\sigma)_\Gamma(A)
}
\]
As a composition of isomorphisms it is an isomorphism so $GF(\Gamma\cext A)$ is also a context comprehension of $GF\Gamma$ and $(\tau\sigma)_\Gamma(A)$.
We must still check that the corresponding projections are those required by the definition. It is obvious for the first projection:
\begin{eqnarray*}
\p \rho^G_{F\Gamma, \sigma_\Gamma A} G(\rho^F_{\Gamma, A}) 	&=& G(\p)G(\rho^F_{\Gamma, A})\\
								&=& GF\p
\end{eqnarray*}
But more intricate for the second.
\begin{eqnarray*}
\q[\rho^G_{F\Gamma, \sigma_\Gamma A} G(\rho^F_{\Gamma, A})]	&=_1& \{(\theta^G_{\sigma_\Gamma A, \p})^{-1}\}(\tau_{F\Gamma\cext \sigma_\Gamma A}^{\sigma_\Gamma(A)[\p]}(\q))[G(\rho^F_{\Gamma, A})]\\
								&=_2& \{(\theta^G_{\sigma_\Gamma A, \p})^{-1}\}(\{(\theta^G_{\sigma_\Gamma(A)[\p], \rho^F_{\Gamma, A}})^{-1}\}(\tau_{F(\Gamma\cext A)}^{\sigma_\Gamma(A)[F\p]}(\q[\rho^F_{\Gamma, A}])))\\
								&=_3& \{(\theta^G_{\sigma_\Gamma A,F\p})^{-1}\}(\tau_{F(\Gamma\cext A)}^{\sigma_\Gamma(A)[F\p]}(\q[\rho^F_{\Gamma, A}]))\\
								&=_4& \{(\theta^G_{\sigma_\Gamma A,F\p})^{-1}\}(\tau_{F(\Gamma\cext A)}^{\sigma_\Gamma(A)[F\p]}(\{(\theta^F_{A, \p})^{-1}\}(\sigma_{\Gamma\cext A}^{A[\p]}(\q))))\\
								&=_5& \{(\theta^G_{\sigma_\Gamma A,F\p})^{-1}\}(\{\indexed{\tau}_{F(\Gamma\cext A)}((\theta_{A, \p}^F)^{-1})\}(\tau_{F(\Gamma\cext A)}^{\sigma_{\Gamma\cext A}(A[\p])}(\sigma_{\Gamma\cext A}^{A[\p]}(\q))))\\
								&=_6& \{\theta_{A, [p]}^{-1}\}((\tau\sigma)_{\Gamma\cext A}^{A[\p]}(q))
\end{eqnarray*}
Where (1) is preservation of the second projection by $\rho^G$, (2) is preservation of substitution on terms, (3) is coherence for $\theta^G$, (4) is preservation
of the second projection by $\rho^F$, (5) is Lemma \ref{comm_sigma_theta_appendix} and (6) is by definition of $\theta$ and $\tau\sigma$.
\end{itemize}
Finally note that, as can be checked by unfolding the definitions, we have for all context $\Gamma$ in $\C$
\[
(\indexed{\tau\sigma})_\Gamma = \indexed{\tau}_{F\Gamma} \circ \indexed{\sigma}_\Gamma
\]
Hence the necessary coherence and naturality conditions amounts to the stability of pseudonatural transformations under composition.
\end{proof}

\begin{lemma}
Any pseudo cwf-morphism $(F,\sigma)$ from $(\C, T)$ to $(\C', T')$ where both cwfs support $\Sigma$-types automatically preserves them,
in the sense that
\[
\sigma_\Gamma(\Sigma(A, B)) \iso \Sigma(\sigma_\Gamma(A), \sigma_{\Gamma\cext A}(B)[\rho_{\Gamma, A}^{-1}])
\]
\end{lemma}
\begin{proof}
We exploit the fact that in any cwf $(\C, T)$ with $\Sigma$-types we have $\Gamma\cext \Sigma(A, B) \iso \Gamma\cext A\cext B$ (obvious).
Therefore:
\begin{eqnarray*}
F\Gamma\cext \sigma_\Gamma(\Sigma(A, B))&\iso& F(\Gamma\cext \Sigma(A, B))\\
                                        &\iso& F(\Gamma\cext A \cext B)\\
                                        &\iso& F(\Gamma\cext A)\cext \sigma_{\Gamma\cext A}(B)\\
                                        &\iso& F\Gamma\cext \sigma_\Gamma(A) \cext \sigma_{\Gamma\cext A}(B)[\rho_{\Gamma, A}^{-1}]\\
                                        &\iso& F\Gamma\cext \Sigma(\sigma_\Gamma(A), \sigma_{\Gamma\cext A}(B)[\rho_{\Gamma, A}^{-1}])
\end{eqnarray*}
It is easy to see that the resulting morphism is a type isomorphism.
\end{proof}

\begin{lemma}
Preservation of democracy, identity types and $\Pi$-types are all stable under composition.
\end{lemma}
\begin{proof}
This is trivial for identity types and $\Pi$-types (just apply the hypothesis on both pseudo cwf-morphism). For democracy, we must check that
the isomorphism $d_\Gamma^{GF}: GF\nilc \cext (\tau\sigma)_{\nilc}(\bar{\Gamma}) \to GF\nilc \cext \bar{GF\Gamma}[\subst{}]$ defined by:
\[
d_\Gamma^{GF} = 
\xymatrix@R=0pt@C=40pt{
GF\nilc\cext \tau_{F\nilc}(\sigma_{\nilc}(\bar{\Gamma}))
        \ar[r]^{(\rho^G_{F\nilc, \sigma_{\nilc}\bar{\Gamma}})^{-1}}&
F(F\nilc\cext \sigma_{\nilc}\bar{\Gamma})
        \ar[r]^{G(d^F_\Gamma)}&
G(F\nilc\cext \bar{F\Gamma}[\subst{}])\\
~        \ar[r]^{G(\subst{\subst{}, \q})}&
G(\nilc\cext \bar{F\Gamma})
	\ar[r]^{\rho^G_{\nilc, \bar{F\Gamma}}}&
G\nilc\cext \tau_{\nilc}\bar{F\Gamma}\\
&~      \ar[r]^{d^G_{F\Gamma}}&
G\nilc\cext \bar{GF\Gamma}[\subst{}]
}
\]
The coherence law can then be checked by a simple diagram chasing.
\end{proof}

\subsection{Properties and composition of pseudo cwf-transformations}

Let us just remark that if $(F, \sigma)$ and $(G, \tau)$ are pseudo cwf-morphisms from $(\C, T)$ to $(\D, T')$, 
pairs $(\phi, m)$ where $\phi: F \natto G$ is a natural transformation and $m: (\indexed{T'}\star\phi)\circ \indexed{\sigma} \modto \indexed{\tau}$
is a modification (where $\indexed{T'}\star\phi$ denotes the \emph{vertical composition} of the natural transformations $\id_{\indexed{T'}}$ and $\phi$)
exactly correspond to pseudo cwf-transformations from $(F, \sigma)$ to $(G, \tau)$ (as can be checked by unfolding the definition of a modification).

It is folklore that there is a $2$-category $\Ind$ of indexed categories over arbitrary bases, which objects are pairs $(\C, \indexed{T})$
(where $\C$ is a category and $\indexed{T}: \C^{op}\to \Cat$ is a pseudofunctor), $1$-cells are pairs $(F, \indexed{\sigma}) : (\C, \indexed{T}) \to (\D, \indexed{T'})$
(where $F: \C \to \D$ is a functor and $\indexed{\sigma} : \indexed{T} \to \indexed{T'} F$ is a pseudonatural transformation) and $2$-cells
are pairs $(\phi, m): (F, \indexed{\sigma}) \to (G, \indexed{\tau}):(\C, \indexed{T}) \to (\D, \indexed{T'})$ (where $\phi : F \natto G$ is a natural transformation
and $m: (\indexed{T'}\star \phi)\circ \indexed{\sigma} \modto \indexed{\tau}$ is a modification).

Here we rely on $\Ind$ to define vertical and horizontal composition of pseudo cwf-transformations, so we get various $2$-categories of
cwfs supporting structure, structure-preserving pseudo cwf-morphisms and pseudo cwf-transformations between them, which
can all be seen as sub-$2$-categories of $\Ind$. In particular, we will be interested in the $2$-category $\CWFFL$ of cwfs supporting
democracy, $\Sigma$-types and identity types, pseudo cwf-morphisms preserving democracy and identity types and pseudo cwf-transformations.
We will also be interested in the $2$-category $\CWFLCC$ where $\Pi$-types are additionally supported and preserved.

\section{Proofs of Section 4}

For this section, the main thing to prove is that the inverse image is preserved (up to isomorphism) by structure-preserving
pseudo cwf-morphisms:
This relies mostly on long and intricate calculations involving all the components of pseudo cwf-morphisms.
Let us first prove a few preliminary lemmas.

\begin{lemma}[Propagation of isomorphisms]
Isomorphisms propagate through types in several different ways. Suppose that you have $A, A'\in Type(\Gamma)$, $B, B'\in Type(\Gamma\cext A)$, then
\begin{itemize}
\item[(1)] If $B\iso B'$, then $\Sigma(A, B) \iso \Sigma(A, B')$
\item[(2)] If $A\iso_\theta A'$, then $\Sigma(A, B) \iso \Sigma(A', B[\theta^{-1}])$
\item[(3)] If $A\iso_\theta A'$ and $a, a' \in \Gamma\vdash A$, then $\Id_A(a, a') \iso \Id_{A'}(\applyopen{\theta}{(a)}, \applyopen{\theta}{(a')})$
\end{itemize}
\label{prop_isos_appendix}
\end{lemma}
\begin{proof}
$(1)$ is obvious, since $\Gamma\cext \Sigma(A, B)$ is isomorphic to $\Gamma\cext A\cext B$. For $(2)$, we give the following two isomorphisms:
\begin{eqnarray*}
\subst{p, \pair(\q[\theta\subst{\p, \pi_1(\q)}], \pi_2(\q))} &:& \Gamma\cext \Sigma(A, B) \to \Gamma \cext \Sigma(A', B[\theta^{-1}])\\
\subst{p, \pair(\q[\theta^{-1}\subst{\p, \pi_1(\q)}], \pi_2(\q))} &:& \Gamma \cext \Sigma(A', B[\theta^{-1}]) \to \Gamma\cext \Sigma(A, B)
\end{eqnarray*}
A simple calculation shows that they typecheck and that they are inverse of one another. It is obvious that they are isomorphisms of types.
$(3)$ is also obvious since by extensionality, $\subst{\p, \refl}$ typechecks in both directions and is its own inverse.
\end{proof} 

\begin{lemma}[Preservation of inverse image]
Let $(\C, T)$ and $(\D, T')$ be cwfs supporting democracy, $\Sigma$-types and identity types and let $(F, \sigma)$ be a pseudo cwf-morphism
preserving them. Moreover, suppose that $\delta : \Delta \to \Gamma$ is a morphism in $\C$, then
there is an isomorphism in $\D$:
\begin{eqnarray*}
F(\Gamma\cext \Inverse(\delta)) &\iso& F\Gamma\cext \Inverse(F\delta)
\end{eqnarray*}

\label{pres_inverse_appendix}
\end{lemma}
\begin{proof}
Exploiting Lemma \ref{prop_isos_appendix} and preservation of substitution on types and terms, a careful (but rather straightforward) calculation allows
to derive the following type isomorphism:
\[
\sigma_\Gamma(\Inverse(\delta)) \iso \Sigma(\bar{F\Delta}[\subst{}], \Id_{\bar{F \Gamma}[\subst{}]}(C(\sigma_{\Gamma\cext \bar{\Delta}[\subst{}]}^{\bar{\Gamma}[\subst{}]}(\bar{\delta}[\subst{\subst{}, \q}])), C(\sigma_{\Gamma\cext \bar{\Delta}[\subst{}]}^{\bar{\Gamma}[\subst{}]}(\q[\gamma_\Gamma \p]))))
\]
where $C(-)$ is an invertible context given by:
\[
C(M) = \applyopen{\indexed{T'}{(\iota \subst{})}{(d_\Gamma)} \theta_{\bar{\Gamma}, \subst{}}^{-1}}{(M)}[\rho_{\Gamma, \bar{\Delta}[\subst{}]} \theta_{\bar{\Delta}, \subst{}} \indexed{T'}{(\iota \subst{})}{(d_\Delta^{-1})}]
\]
Here, $\iota$ denotes the inverse of the terminal morphism $\subst{} : F\nilc\to \nilc$ whose existence is asserted by the definition of a pseudo cwf-morphism.
Hence, the goal of the remaining part of this proof will be to show the following equalities:
\begin{eqnarray}
\sigma_{\Gamma\cext \bar{\Delta}[\subst{}]}^{\bar{\Gamma}[\subst{}]}(\bar{\delta}[\subst{\subst{}, \q}]) &=& C^{-1}(\bar{F\delta}[\subst{\subst{}, \q}])\label{one}\\
\sigma_{\Gamma\cext \bar{\Delta}[\subst{}]}^{\bar{\Gamma}[\subst{}]}(\q[\gamma_\Gamma \p]) &=& C^{-1}(\q[\gamma_{F \Gamma} \p])\label{two}
\end{eqnarray}

Let us focus on \eqref{one}. Using preservation of substitution on terms, coherence of $\theta$ and the basic computation laws in cwfs,
we derive:

\begin{eqnarray*}
\sigma_{\Gamma\cext \bar{\Delta}[\subst{}]}^{\bar{\Gamma}[\subst{}]}(\bar{\delta}[\subst{\subst{}, \q}]) &=& 
        \sigma_{\Gamma\cext \bar{\Delta}[\subst{}]}^{\bar{\Gamma}[\p][\subst{\subst{},\q}]}(\bar{\delta}[\subst{\subst{}, \q}])\\
&=&     \applyopen{\theta_{\bar{\Gamma}[\p], \subst{\subst{},\q}}}{(\sigma_{\nilc\cext \bar{\Delta}}^{\bar{\Gamma}[\p]}(\bar{\delta})[F(\subst{\subst{},\q})])}\\
&=&     \applyopen{\theta_{\bar{\Gamma}, \subst{}}}{(\applyopen{\indexed{T'}{(\subst{\subst{},\q})}{(\theta_{\bar{\Gamma}, \p}^{-1})}}{(\sigma_{\nilc\cext \bar{\Delta}}^{\bar{\Gamma}[\p]}(\bar{\delta})[F(\subst{\subst{},\q})])})}\\
&=&     \applyopen{\theta_{\bar{\Gamma}, \subst{}}}{(\applyopen{\theta_{\bar{\Gamma}, \p}^{-1}}{(\sigma_{\nilc\cext \bar{\Delta}}^{\bar{\Gamma}[\p]}(\bar{\delta}))}[F(\subst{\subst{},\q})])}
\end{eqnarray*}

Let us now focus on $\sigma_{\nilc\cext \bar{\Delta}}^{\bar{\Gamma}[\p]}(\bar{\delta})$, to see how terms created from substitution using democracy
are transformed by the action of the cwf-morphism. Here, we are only going to use the coherence of $\theta$, preservation of $\q$ and the basic
computation laws in cwfs.

\begin{eqnarray*}
\sigma_{\nilc\cext \bar{\Delta}}^{\bar{\Gamma}[\p]}(\bar{\delta}) &=& \sigma_{\nilc\cext \bar{\Delta}}^{\bar{\Gamma}[\p]}(\q[\gamma_\Gamma \delta \gamma_\Delta^{-1}])\\
&=& \sigma_{\nilc\cext \bar{\Delta}}^{\bar{\Gamma}[\p][\gamma_\Gamma \delta \gamma_\Delta^{-1}]}(\q[\gamma_\Gamma \delta \gamma_\Delta^{-1}])\\
&=& \applyopen{\theta_{\bar{\Gamma}[\p], \gamma_\Gamma\delta \gamma_\Delta^{-1}}}{(\sigma_{\nilc\cext \bar{\Gamma}}^{\bar{\Gamma}[\p]}(\q)[F(\gamma_\Gamma\delta \gamma_\Delta^{-1})])}\\
&=& \applyopen{\theta_{\bar{\Gamma}[\p], \gamma_\Gamma\delta \gamma_\Delta^{-1}}}{(\applyopen{\theta_{\bar{\Gamma}, \p}}{(\q[\rho_{\nilc, \bar{\Gamma}}])}[F(\gamma_\Gamma\delta \gamma_\Delta^{-1})])}\\
&=& \applyopen{\theta_{\bar{\Gamma}, \p}}{(\applyopen{\indexed{T'}{(F(\gamma_\Gamma \delta \gamma_\Delta^{-1}))}{(\theta_{\bar{\Gamma}, \p}^{-1})}}{(\applyopen{\theta_{\bar{\Gamma}, \p}}{(\q[\rho_{\nilc, \bar{\Gamma}}])}[F(\gamma_\Gamma\delta \gamma_\Delta^{-1})])})}\\
&=& \applyopen{\theta_{\bar{\Gamma}, \p}}{(\q[\rho_{\nilc, \bar{\Gamma}} F(\gamma_\Gamma \delta \gamma_\Delta^{-1})])}
\end{eqnarray*}

Using preservation of democracy and the terminal object, we can now conclude:
\begin{eqnarray*}
\sigma_{\Gamma\cext \bar{\Delta}[\subst{}]}^{\bar{\Gamma}[\subst{}]}(\bar{\delta}[\subst{\subst{}, \q}]) &=&
    \applyopen{\theta_{\bar{\Gamma}, \subst{}}}{(\q[\rho_{\nilc, \bar{\Gamma}} F(\gamma_\Gamma \delta \gamma_\Delta^{-1}\subst{\subst{},\q})])}\\
&=& \applyopen{\theta_{\bar{\Gamma}, \subst{}}}{(\q[d_r^{-1}\subst{\iota \p, \q} \gamma_{F \Gamma} F\delta \gamma_{F\Delta}^{-1}\subst{\subst{},\q}\indexed{T'}{(\iota \subst{})}{(d_\Delta)}\theta_{\bar{\Delta}, \subst{}}^{-1} \rho_{\Gamma, \bar{\Delta}[\subst{}]}])}\\
&=& \applyopen{\theta_{\bar{\Gamma}, \subst{}} \indexed{T'}{(\iota\subst{})}{(d_\Gamma^{-1})}}{(\bar{F\delta}[\subst{\subst{},\q}])}[\indexed{T'}{(\iota \subst{})}{(d_\Delta)}\theta_{\bar{\Delta}, \subst{}}^{-1} \rho_{\Gamma, \bar{\Delta}[\subst{}]}]\\
&=& C^{-1}(\bar{F\delta}[\subst{\subst{},\q}])
\end{eqnarray*}

We get the required expression. The case of Equation \eqref{two} is similar but less intricate, so we skip the details. 
\end{proof}

\begin{lemma}[Propagation of isomorphisms under $\Pi$]
Suppose that we have $A, A'\in Type(\Gamma)$ and $B, B'\in Type(\Gamma\cext A)$, then
\begin{enumerate}
\item If $B\iso_\theta B'$, then $\Pi(A, B) \iso \Pi(A, B')$
\item If $A\iso_\theta A'$, then $\Pi(A, B) \iso \Pi(A', B[\theta^{-1}])$
\end{enumerate}
\label{prop_pi_appendix}
\end{lemma}
\begin{proof}
(1). Using $\theta$ and the combinators of $\Pi$-types, it is straightforward to build an morphism:
\[
\subst{\p, \lambda(\applyopen{\indexed{T}{(\subst{\p\p,\q})}{(\theta)}}{(\ap(\q[\p], \q))})} : \Gamma\cext \Pi(A, B) \to \Gamma\cext \Pi(A, B')
\]
Its inverse is the corresponding expression with $\theta^{-1}$ in place of $\theta$.
(2). Likewise, the following expression provides the required isomorphism:
\[
\lambda(\ap(\q[\p\p], \applyopen{\indexed{T}{(\p\p)}{(\theta^{-1})}}{(\q)})) : \Gamma\cext \Pi(A, B) \to \Gamma\cext \Pi(A', B[\theta^{-1}])
\]
\end{proof}

\begin{proposition}
Let $(F, \sigma)$ be a pseudo cwf-morphism between $(\C, T)$ and $(\C', T')$ supporting $\Sigma$-types and democracy. Then:
\begin{itemize}
\item If $(\C, T)$ and $(\C', T')$ both support identity types and $(F, \sigma)$ preserves them, then $F$ preserves finite limits.
\item If $(\C, T)$ and $(\C', T')$ both support $\Pi$-types and $(F, \sigma)$ preserves them, then $F$ preserves local exponentiation.
\end{itemize}
\label{preservation_structure_appendix}
\end{proposition}
\begin{proof}
Since finite limits and local exponentiation can be defined using $\sigma$-types, $\Pi$-types and the inverse image type, their preservation
by $F$ directly boils down to the Lemmas \ref{prop_isos_appendix}, \ref{pres_inverse_appendix} and \ref{prop_pi_appendix}.
\end{proof}

\section{Proofs of Section 5}

\subsection{Action on $0$-cells}

This section is the exact analogue for cwfs of Hofmann's work \cite{hofmann:csl} with cwas. To make this paper self-contained we will give the full details of the construction.
However, we will skip some of the proofs whenever they are not significantly different from the case of cwas, for which we refer to \cite{hofmann:csl}.

\subsubsection{Base cwf structure.}

We will start by proving that given any category $\C$ with a terminal object, we can equip $\C$ with cwf structure. This means that we have to define a functor $T_\C : \C \rightarrow \Fam$
(in other terms types, terms, and substitution on both of them), and a context comprehension operation.

\paragraph{Types.} A \emph{type} over $\Gamma$ is a \emph{functorial family}, \emph{i.e.} a functor $\hf{A}:\C/\Gamma\to \arrow{\C}$ such that:
\begin{itemize}
\item[(i)] $cod\circ \hf{A} = dom$
\item[(ii)] If
$\xymatrix@R=5pt@C=5pt{
\Omega\ar[dr]_{\delta\alpha}\ar[rr]^{\alpha}&&\Delta\ar[dl]^{\delta}\\
&\Gamma
}$
is a morphism in $\mathbb{C}/\Gamma$, $\hf{A}(\alpha)$ is a pullback square, with the naming convention below:
\[
\xymatrix{
dom(\hf{A}(\delta\circ\alpha))\ar[r]^{\hf{A}(\delta,\alpha)}
                         \ar[d]_{\hf{A}(\delta \alpha)} &
dom(\hf{A}(\delta))\ar[d]^{\hf{A}(\delta)}\\
\Omega\ar[r]_\alpha & \Delta
}
\]
\end{itemize}
Let $Type(\Gamma)$ denote the set of functorial families over $\Gamma$.

\paragraph{Remark.}
The functoriality of $\hf{A}$ means that the assignment of $\hf{A}(\delta,\alpha)$ satisfies the following equations:
\begin{itemize}
\item $\hf{A}(\delta,\id_\Delta) = \id_{dom(\hf{A}(\delta))}$
\item $\hf{A}(\delta,\alpha \beta) = \hf{A}(\delta,\alpha)\hf{A}(\delta\alpha,\beta)$
\end{itemize}

\paragraph{Terms.} Let $\Gamma\in \mathbb{C}$, and $\hf{A}\in Type(\Gamma)$. We define
$\Gamma\cext \hf{A} = dom(\hf{A}(\id_\Gamma))$. (This will later give us context comprehension.)
Then, a term $a:\Gamma\vdash \hf{A}$ is a morphism $a:\Gamma\to \Gamma\cext \hf{A}$ such that
$\hf{A}(\id_\Gamma) a = \id_\Gamma$.

\paragraph{Substitution in types.}
Let $\gamma:\Delta\to \Gamma$  in $\mathbb{C}$ and $\hf{A}\in Type(\Gamma)$. We define $\hf{A}[\gamma] \in Type(\Delta)$ as follows.
\begin{eqnarray*}
\hf{A}[\gamma](\delta) &=& \hf{A}(\gamma\delta)\\
\hf{A}[\gamma](\delta,\alpha) &=& \hf{A}(\gamma\delta,\alpha)
\end{eqnarray*}
where $\delta:\Omega\to \Delta$ and $\alpha:\Xi\to \Omega$.
We check that $\hf{A}[\gamma]$ satisfies the two conditions for types:
\begin{itemize}
\item[(i)] $cod\circ \hf{A}[\gamma] = dom$.
\begin{eqnarray*}
cod\circ \hf{A}[\gamma] (\delta) &=& cod\circ \hf{A}(\gamma \delta)\\
                            &=& dom(\gamma\delta)\\
                            &=& dom(\delta)
\end{eqnarray*}
\item[(ii)] The image of the morphism
$\xymatrix@R=5pt@C=5pt{
\Xi\ar[dr]_{\delta \alpha}\ar[rr]^{\alpha}&&\Omega\ar[dl]^{\delta}\\
&\Delta
}$ in $\mathbb{C}/\Delta$  is
\[
\hf{A}[\gamma](\delta)~~~=
\raisebox{25pt}{
\xymatrix@R=30pt@C=30pt{
dom(\hf{A}(\gamma\delta\alpha))\ar[r]^{\hf{A}(\gamma\delta,\alpha)}
                         \ar[d]_{\hf{A}(\gamma \delta \alpha)} &
dom(\hf{A}(\gamma \delta))\ar[d]^{\hf{A}(\gamma\delta)}\\
\Xi\ar[r]_\alpha & \Omega
}}
\]
This is a pullback square by property (ii) of $\hf{A}$.
\end{itemize}

\paragraph{Substitution in terms.}
Let $\delta:\Delta\to \Gamma$, and $a:\Gamma\vdash \hf{A}$, \emph{i.e.} $a:\Gamma\to \Gamma\cext \hf{A}$ such that
$\hf{A}(\id_\Gamma)\circ a = \id_\Gamma$. Then $a[\delta]$ is defined as the unique mediating arrow in the following diagram:
\[
\xymatrix@R=50pt@C=50pt{
\Delta  \ar@/^/[drr]^{a\circ \delta}
        \ar@/_/[ddr]_{\id_\Delta}
        \ar@{.>}[dr]|{a[\delta]}\\
&\Delta\circ \hf{A}[\delta]
        \ar[r]^{\hf{A}(\id_\Gamma,\delta)}
        \ar[d]^{\hf{A}[\delta](\id_\Delta)} 
        &\Gamma\cext \hf{A}
                \ar[d]^{\hf{A}(\id_\Gamma)}\\
&\Delta \ar[r]_\delta&\Gamma
}
\]
It is a term of type $\hf{A}[\delta]$ by commutativity of the left triangle.

\paragraph{Functoriality.}
Since substitution in types is defined by composition, the cwf-laws for it follow immediately.
It is also immediate that $a[\id_\Gamma] = a$ since the defining pullback must be the identity in $\arrow{\C}$. It
remains to show that if we have $\delta_1:\Delta'\to \Delta$ and $\delta_2:\Delta\to \Gamma$, then
$a[\delta_2\circ \delta_1] = a[\delta_2][\delta_1]$. Consider the following diagram for pullback composition:
\[
\xymatrix@R=50pt@C=50pt{
\Delta' \ar@/_2pc/[ddr]_{\id_{\Delta'}}
        \ar@/^2pc/[drrr]^{a\circ\delta_2\circ\delta_1}
        \ar[dr]|{a[\delta_2][\delta_1]}
\\
&\Delta'\cext \hf{A}[\delta_2\circ\delta_1]
        \ar[d]|{\hf{A}[\delta_2\circ\delta_1](id_{\Delta'})}
        \ar[r]^{\hf{A}[\delta_2](\id_{\Delta},\delta_1)} 
&\Delta\cext \hf{A}[\delta_2]   
        \ar[d]|{\hf{A}[\delta_2](\id_\Delta)}
        \ar[r]^{\hf{A}(\id_\Gamma,\delta_2)}
&\Gamma\cext \hf{A}
        \ar[d]^{\hf{A}(\id_\Gamma)}
\\
&\Delta'        \ar[r]_{\delta_1}
&\Delta         \ar[r]_{\delta_2}
&\Gamma
}
\]
The external pullback square is obtained by substitution of $\delta_2\circ\delta_1$. By definition of
substitution in terms, $a[\delta_2\circ\delta_1]$ is the unique mediating arrow. But the external square is equal (not only up to isomorphism) to the composition of the smaller squares, because of the functoriality conditions for
$\hf{A}$, more precisely the fact that $\hf{A}(s,\alpha\circ \beta) = \hf{A}(s,\alpha)\circ \hf{A}(s\circ \alpha,\beta)$.
This implies that $a[\delta_2][\delta_1]$ also makes the two triangles commute. Hence $a[\delta_2][\delta_1] = a[\delta_2\circ\delta_2]$
by uniqueness of the mediating arrow.

Putting all this together, we now have built a functor $T_\C: \C^{op} \to \Fam$. We still have to define context comprehension.

\paragraph{Context comprehension.} Let $\Gamma\in \mathbb{C}$, and $\hf{A}\in Type(\Gamma)$. As mentioned above, we define:
\[
\Gamma\cext \hf{A} = dom(\hf{A}(\id_\Gamma))
\]
The first projection is $\p_{\hf{A}} = \hf{A}(\id_\Gamma) : \Gamma\cext \hf{A} \to \Gamma$. The second projection $\q_{\hf{A}}$ is defined
as the unique mediating arrow of the following pullback diagram:
\[
\xymatrix@R=30pt@C=50pt{
\Gamma\cext \hf{A}
        \ar@/_1.5pc/[ddr]_{\id_{\Gamma\cext \hf{A}}}
        \ar@/^1.5pc/[drr]^{\id_{\Gamma\cext \hf{A}}} 
        \ar@{.>}[dr]|{\q_{\hf{A}}} &&\\
&\Gamma\cext \hf{A}\cext \hf{A}[\p_{\hf{A}}]
        \ar[r]^{\hf{A}(\id_\Gamma,\p_{\hf{A}})}
        \ar[d]^{\hf{A}[\p_{\hf{A}}](\id_{\Gamma\cext \hf{A}})}
&\Gamma\cext \hf{A}
        \ar[d]^{\hf{A}(\id_\Gamma)}\\
&\Gamma\cext \hf{A}
        \ar[r]_{\p_{\hf{A}}}
&\Gamma
}
\]
Suppose now we have $\delta:\Delta \to \Gamma$ and $a:\Delta \vdash \hf{A}[\delta]$. By definition of terms we have in fact
$a:\Delta \to \Delta\cext \hf{A}[\delta]$. We define:
\[
\ext{\delta,a} = \hf{A}(id_\Gamma,\delta)\circ a : \Delta \to \Gamma\cext \hf{A}
\]
We must prove that these definitions satisfy the cwf-laws for context comprehension.
\[
\begin{array}{rclr}
\p_{\hf{A}}\circ \ext{\delta,a}  &=& \p_{\hf{A}} \circ \hf{A}(\id_\Gamma,\delta)\circ a&\text{def.of $\ext{\delta,a}$}\\
                                &=& \hf{A}(\id_\Gamma) \circ \hf{A}(\id_\Gamma,\delta)\circ a&\text{def.~of $\p_{\hf{A}}$}\\
                                &=& \delta \circ \hf{A}[\delta](\id_\Delta)\circ a&\text{comm.~of pullback square}\\
                                &=& \delta &a\ \text{is\ a\ term}\\
\end{array}
\]
Proving the equation $\q_{\hf{A}}[\ext{\delta,a}] = a$ is a bit more involved. Let us first prove the
following lemma, stating (intuitively) that the action of $\q_{\hf{A}}$ is to duplicate the last element of the context.
\begin{lemma}Let $\hf{A}\in Type(\Gamma)$, $\delta:\Delta\to \Gamma$, and $a:\Delta\vdash \hf{A}[\delta]$. Then
$\q_{\hf{A}}\circ \ext{\delta,a} = \ext{\ext{\delta,a},a}$.
\end{lemma}
\begin{proof}
Consider the following pullback diagram:
\[
\xymatrix@R=50pt@C=50pt{
\Delta  \ar@/^/[drr]^{\ext{\delta,a}}
        \ar@/_/[ddr]_{\ext{\delta,a}}
        \ar@{}[drr]|{\raisebox{-20pt}{(1)}}
        \ar@{}[ddr]|{\raisebox{0pt}{~~~~(2)}}
        \ar@/_/[dr]|{\ext{\ext{\delta,a},a}}
        \ar@/^/[dr]|{\q_{\hf{A}}\circ \ext{\delta,a}}&&\\
&\Gamma\cext \hf{A}\cext \hf{A}[\p_{\hf{A}}]
        \ar[r]^{\hf{A}(\id_\Gamma,\p_{\hf{A}})}
        \ar[d]^{\hf{A}[\p_{\hf{A}}](\id_{\Gamma\cext \hf{A}})}
&\Gamma\cext \hf{A}
        \ar[d]^{\hf{A}(\id_\Gamma)}\\
&\Gamma\cext \hf{A}
        \ar[r]_{\p_{\hf{A}}}
&\Gamma
}
\]
It is clear that $\q_{\hf{A}}\circ \ext{\delta,a}$ makes (1) and (2) commute, by definition of $\q_{\hf{A}}$. It is also easy to
see that $\ext{\ext{\delta,a},a}$ makes (2) commute, because $\hf{A}[\p_{\hf{A}}](\id_{\Gamma\cext \hf{A}}) = \p_{\hf{A}[\p_{\hf{A}}]}$
and by the property of the first projection. We prove now that $\ext{\ext{\delta,a},a}$ also makes (1) commute:
\[
\begin{array}{rclr}
\hf{A}(\id_\Gamma,\p_{\hf{A}})\circ \ext{\ext{\delta,a},a}&=&
                \hf{A}(\id_\Gamma,\p_{\hf{A}})\circ \hf{A}[\p_{\hf{A}}](id_{\Gamma\cext \hf{A}},\ext{\delta,a})\circ a
                                &\text{def. of $\ext{\ext{\delta,a},a}$}\\
&=&\hf{A}(\id_\Gamma,\p_{\hf{A}})\circ \hf{A}(\p_{\hf{A}},\ext{\delta,a})\circ a
                                &\text{def. of $\hf{A}[\p_{\hf{A}}]$}\\
&=&\hf{A}(\id_\Gamma,\p_{\hf{A}}\circ \ext{\delta,a})\circ a
                                &\text{funct. or $\hf{A}(s,\delta)$}\\
&=&\hf{A}(\id_\Gamma,\delta)\circ a
                                &\text{property of $\p_{\hf{A}}$}\\
&=&\ext{\delta,a}
                                &\text{def. of $\ext{\delta,a}$}
\end{array}
\]
This concludes the proof.
\end{proof}
From this lemma we deduce that $\q_{\hf{A}}[\ext{\delta,a}] = a$ in the following way. Consider the
pullback diagram:
\[
\xymatrix@R=50pt@C=70pt{
\Delta  \ar@/^/[drr]^{\q_{\hf{A}}\circ \ext{\delta,a}}
        \ar@/_/[ddr]_{\id_\Delta}
        \ar@/_/[dr]|{\q_{\hf{A}}[\ext{\delta,a}]}
        \ar@/^/[dr]|a&&\\
&\Delta\cext \hf{A}[\delta]
        \ar[r]^{\hf{A}[\p_{\hf{A}}](\id_{\Gamma\cext \hf{A}},\ext{\delta,a})}
        \ar[d]^{\p_{\hf{A}[\delta]}}
&\Gamma\cext \hf{A}\cext \hf{A}[\p_{\hf{A}}]
        \ar[d]^{\hf{A}[\p_{\hf{A}}](\id_{\Gamma\cext \hf{A}})}\\
&\Delta \ar[r]_{\ext{\delta,a}}
&\Gamma\cext \hf{A}
}
\]
This diagram is an instance of the pullback used for the definition of substitution in terms. Hence,
both triangles commute for $\q_{\hf{A}}[\ext{\delta,a}]$. The left triangle commutes for $a$ since a is term of type $\hf{A}[\delta]$.
The right triangle commutes because of the lemma above, since by definition
$\hf{A}[\p_{\hf{A}}](\id_{\Gamma\cext \hf{A}},\ext{\delta,a}) = \ext{\ext{\delta,a},a}$. Thus, by uniqueness of the mediating arrow $\q_{\hf{A}}[\ext{\delta,a}] = a$.
This concludes the cwf construction, hence the proof of the following proposition.

\begin{proposition}
Let $\mathbb{C}$ be a category with terminal object, then we can extend $\C$ to a cwf $(\C, T_\C)$.
\label{cwf_appendix}
\end{proposition}

\subsubsection{Democracy.}The cwf $(\C, T_\C)$ is democratic: the idea is that each context $\Gamma$ is represented by any functorial family
having its terminal projection $\ext{}: \Gamma \to I$ as display map. We can easily build such a functorial family by
$\bar{\Gamma} = \hat{\ext{}}\in Type(I)$. We have then $I \cext \bar{\Gamma} = dom(\hat{\ext{}}(\id)) = \Gamma$, thus the isomorphism
between them is trivial.

\begin{proposition}
If $\C$ is a category with a terminal object, then the cwf $(\C, T_\C)$ is democratic.
\end{proposition}

\subsubsection{$\Sigma$-types.}
For the sake of completeness we recall the definitions, but we
refer the reader to \cite{hofmann:csl} for some of the proofs, in particular when the distinction between cwas and cwfs does not change anything.

\paragraph{Formation.} Let $A\in Type(\Gamma)$ and $B\in Type(\Gamma\cext A)$. At each $s : \Delta \to \Gamma$, the image of
$\Sigma(A, B)$ is given by the composition of the images of $A$ and $B$. More formally, we define:
\begin{eqnarray*}
\Sigma(A,B)(s) &=& A(s) \circ B(A(\id,s))\\
\Sigma(A,B)(s, \alpha) &=& B(A(\id,s), A(s, \alpha))
\end{eqnarray*}
The construction of the corresponding pullback square can be illustrated by the following diagram. Intuitively, the chosen pullbacks for $\Sigma(A,B)$
are directly obtained by composition the chosen pullbacks for $A$ and for $B$.
\[
\xymatrix@R=30pt@C=80pt{
~       \ar[r]^{B(A(\id,s), A(s, \alpha))}
        \ar[d]&
~       \ar[r]^{B(\id, A(\id, s))}
        \ar[d]^{B(A(\id,s))}&
\Gamma\cext A\cext B
        \ar[d]^{B(\id)}\\
~       \ar[r]^{A(s, \alpha)}
        \ar[d]_{A(s \alpha)}&
~       \ar[r]^{A(\id, s)}
        \ar[d]^{A(s)}&
\Gamma\cext A
        \ar[d]^{A(\id)}\\
~       \ar[r]_\alpha&
B       \ar[r]_s&
\Gamma
}
\]
It is easy to check that this defines a functor $\Sigma(A, B) : \C/\Gamma\to \arrow{\C}$ and that the necessary equations are satisfied so that we get a type $\Sigma(A, B)\in Type(\Gamma)$.

\paragraph{Introduction.} If $a:\Gamma\vdash A$ and $b:\Gamma\vdash B[\ext{\id, a}]$, then $a : \Gamma\to \Gamma\cext A$ is a section
of $A(id_\Gamma)$ and $b:\Gamma \to \Gamma\cext B[\ext{\id,a}]$ is a section of $B[\ext{\id,a}](id_\Gamma) = B(a)$ as illustrated by
following diagram:
\[
\xymatrix@C=40pt@R=40pt{
\Gamma\cext B[\ext{\id,a}]
        \ar[r]^{B(\id,a)}
        \ar[d]^{B(a)}&
\Gamma\cext A\cext B
        \ar[d]^{B(\id)}\\
\Gamma  \ar@/^/[u]^{b}
        \ar[r]_a&
\Gamma\cext A
}
\]
We  define $pair(a,b) = B(\id,a)\circ b$. It follows that $pair(a,b)$
is a section of $\Sigma(A,B)(\id_\Gamma) = \p_A\circ \p_B$.

\paragraph{Elimination.} Let $c : \Gamma \vdash \Sigma(A, B)$. Thus $c$ is a section of $\p_A\circ \p_B : \Gamma\cext A\cext B \to \Gamma$. We define the first
projection $\pi_1(c) = \p_B\circ c$ which is clearly a section of $\p_A$. The second projection $\pi_2(c)$ is given by the
universal property of the following pullback:

\[
\xymatrix@C=40pt@R=40pt{
\Gamma  \ar@/_/[ddr]_{\id_\Gamma}
        \ar@/^/[drr]^c
        \ar@{.>}[dr]|{\pi_2(c)}\\
&\Gamma\cext B[\ext{\id,a}]
        \ar[r]^{B(\id,a)}
        \ar[d]^{B(a)}&
\Gamma\cext A\cext B
        \ar[d]^{B(\id)}\\
&\Gamma \ar[r]_a&
\Gamma\cext A
}
\]
It is immediate from the diagram that it is a section of $\p_{B[\ext{\id,a}]}$.

\paragraph{Equations.} The equality rules for $\Sigma$-types are proved just as in \cite{hofmann:csl}.

This concludes the proof of the following proposition.

\begin{proposition}
If $\C$ is a category with a terminal object, then $(\C, T)$ supports $\Sigma$-types.
\end{proposition}

\subsubsection{Extensional identity types.} To improve readability, we will now sometimes omit the subscripts of the projections, when they can be recovered from the context.
To build identity types, we require that the base category has finite limits.

\paragraph{Formation rule.}
Let $\Gamma\in \mathbb{C}$, $A\in Type(\Gamma)$, and $a,a':\Gamma\vdash A$. If $s:\Delta\to \Gamma$, we define $\Id_A(a,a')(s)$ as the equalizer
of $a[s]$ and $a'[s]$ (seen as morphisms $\Delta \to \Delta\cext A[s]$). If
$\xymatrix@R=5pt@C=5pt{
\Delta'\ar[dr]_{s'}\ar[rr]^{\delta}&&\Delta\ar[dl]^{s}\\
&\Gamma
}$
is a morphism in $\mathbb{C}/\Gamma$, we define $\Id_A(a,a')(\delta)$ as the upper square in the following diagram:
\[
\xymatrix@R=50pt@C=50pt{
\dom(\Id_A(a,a')(s\delta))
        \ar@{.>}[r]^{\gamma}
        \ar[d]^{\Id_{A}(a,a')(s\delta)}&
\dom(\Id_A(a,a')(s))
        \ar[d]^{\Id_A(a,a')(s)}\\
\Delta'
        \ar[r]^\delta
        \ar@/_/[d]_{a[s\delta]}
        \ar@/^/[d]^{a'[s\delta]}&
\Delta
        \ar@/_/[d]_{a[s]}
        \ar@/^/[d]^{a'[s]}\\
\Delta' \cext A[s\delta]
        \ar[r]_{\ext{\delta \p,\q}}&
\Delta \cext A[s]
}
\]
where $\gamma$ is yet to be defined. For this purpose, and to prove that the obtained square is a pullback, we need the following:

\begin{lemma}In the diagram above, if $f:\dom(f)\to \Delta'$, then $f$ equalizes $a[s\delta]$ and $a'[s\delta]$
iff $\delta f$ equalizes $a[s]$ and $a'[s]$.
\label{eqi_appendix}
\end{lemma}
\begin{proof}
First note that by construction of this cwf, we have the surprising equality $a = \ext{\id_\Gamma,a}$ for any term $a:\Gamma\vdash A$. Indeed,
$\ext{\id_\Gamma,a} = A(\id_\Gamma,\id_\Gamma)\circ a = a$. Thus, we have that
\begin{eqnarray*}
\ext{\delta \p,\q}\circ a[s\delta]        &=& \ext{\delta \p,\q} \circ \ext{\id,a[s\delta]}\\
                                        &=& \ext{\delta,a[s\delta]}\\
                                        &=& \ext{\id_\Delta,a[s]}\circ \delta\\
                                        &=& a[s] \circ \delta
\end{eqnarray*}
For the same reason, we have $\ext{\delta \p,\q}\circ a'[s\delta] = a'[s]\circ \delta$. Suppose now that $f$ equalizes $a[s\delta]$ and
$a'[s\delta]$. Then:
\begin{eqnarray*}
a[s]\circ \delta\circ f         &=& \ext{\delta \p,\q}\circ a[s\delta] \circ f\\
                                &=& \ext{\delta \p,\q}\circ a'[s\delta] \circ f\\
                                &=& a'[s] \circ \delta \circ f
\end{eqnarray*}
Thus as claimed, $\delta f$ equalizes $a[s]$ and $a'[s]$. The same equational reasoning gives the converse implication.
\end{proof}

We use this lemma as follows. We know that $\Id_A(a,a')(s\delta)$ equalizes $a[s\delta]$ and $a'[s\delta]$, thus $\delta \circ \Id_A(a,a')(s\delta)$
equalizes $a[s]$ and $a'[s]$. Thus by the equalizer property, $\delta \circ \Id_A(a,a')(s\delta)$ factors in a unique way through
$\Id_A(a,a')(s)$, and we define $\gamma$ to be the unique morphism. Since we already know that the square commutes, it only remains to prove that it is a pullback square.

Let $h_1:X\to \Delta'$ and $h_2:X\to \dom(\Id_A(a,a')(s))$ be two morphisms which make the outer square commute. Necessarily, $\Id_A(a,a')(s)\circ h_2$
equalizes $a[s]$ and $a'[s]$. Since the outer square commutes, $\delta h_1$ equalizes them as well. By Lemma \ref{eqi_appendix}, $h_1$ equalizes
$a[s\delta]$ and $a'[s\delta]$. Thus it factors uniquely through $\Id_A(a,a')(s\delta)$. Let $h$ be the mediating arrow.
It makes the left triangle commute by the factorisation property and the right triangle commute because $\gamma h$ defines another unique
factorisation of $\delta h_1$ through $\Id_A(a,a')(s)$.

We must check that this construction is functorial. Both conditions (for $\id_s$ and $\delta_1\circ \delta_2$) follow immediately by uniqueness of the factorisation through the equalizer. Thus we have shown that $\Id_A(a,a')\in Type(\Gamma)$.

\paragraph{Reflexivity.} For each $a\in \Gamma\vdash A$, we define the term $\refl_{A,a}:\Gamma\vdash \Id_A(a,a)$ as follows:
\[
\xymatrix{
\Gamma
        \ar@/_/[ddr]_{\id_\Gamma}
        \ar@{.>}[dr]^{\refl_{A,a}}&\\
&dom(\Id_A(a,a)(\id_\Gamma))
        \ar[d]\\
&\Gamma
        \ar@/_/[d]_a
        \ar@/^/[d]^a\\
&\Gamma\cext A
}
\]

\paragraph{Stability under substitution.} First we prove
\[
\Id_A(a,a')[\delta] = \Id_{A[\delta]}(a[\delta],a'[\delta])
\]
It suffices to note that for any $s$, the arrows $\Id_A(a,a')[\delta](s)$ and $\Id_{A[\delta]}(a[\delta],a'[\delta])(s)$
both equalize $a[s\delta]$ and $a'[s\delta]$. The image of arrows is determined uniquely by the factorisation under
this equalizer, hence must also be unchanged.

Then we prove
\[
\refl_{A,a}[\delta] = \refl_{A[\delta],a[\delta]}
\]
This is because $\refl_{A,a}[\delta]$ is a correct factorisation of $\id_\Delta$ through
$\Id_A(a,a')[\delta]=\Id_{A[\delta]}(a[\delta],a'[\delta])$ and $\refl_{A[\delta],a[\delta]}$ is defined as the unique
such factorisation.

\paragraph{Extensionality.}
Here, the judgement $\Gamma\vdash a=a' : A$ means that $a$ and $a'$ are equal morphisms of $\mathbb{C}$.
Suppose we have a term $c:\Gamma\vdash \Id_A(a,a')$.
For the first rule, note that $\Id_A(a,a')(\id_\Gamma)\circ c = \id_\Gamma$, because $c$ is a
term. But $\id_\Gamma$ factors through $\Id_A(a,a')(\id_\Gamma)$. Thus it equalizes $a$ and $a'$, and it follows that $a=a'$.
For the second rule, note that the first rule implies that $a=a'$. Thus $\id_\Gamma$
equalizes $a$ and $a'$ and there is a unique factorisation of $\id_\Gamma$ through $\Id_A(a,a')(\id_\Gamma)$. Since $c$ and
$\refl_{A,a}$ are both such factorisations $c=\refl_{A,a}$.

\begin{proposition}
Let $\C$ be a finitely complete category, then $(\C, T_\C)$ supports identity types.
\end{proposition}

\subsubsection{$\Pi$-types.}
If $\C$ is a lccc, then the cwf $H(\C)$ supports $\Pi$-types. Let $\hf{A}$ be a functorial
family over $\Gamma$ and $\hf{B}$ over $\Gamma\cext \hf{A}$. Then the value of the family $\Pi(\hf{A}, \hf{B})$ at substitution $\delta : \Delta \to \Gamma$
is defined by $\Pi_{\hf{A}(\delta)}(\hf{B}(\hf{A}(\id, \delta)))$, where $\Pi_f$ is the right adjoint of $f^*$ obtained by the lcc structure. If
$\alpha : \Omega \to \Delta$ and $\delta : \Delta \to \Gamma$, we have to define a morphism $\Pi(\hf{A}, \hf{B})(\delta, \alpha)$ yielding a pullback
diagram. For this purpose, first consider the following chain of isomorphisms in $\C/\Omega$:
\begin{eqnarray*}
\Pi_{\hf{A}(\delta\alpha)} \hf{B}(\hf{A}(\id, \delta \alpha)) 
&=& \Pi_{\hf{A}(\delta \alpha)}\hf{B}(\hf{A}(\id, \delta)\hf{A}(\delta, \alpha))\\
&\iso& \Pi_{\hf{A}(\delta \alpha)}(\hf{A}(\delta, \alpha))^*(\hf{B}(\hf{A}(\id, \delta)))\\
&\iso& \alpha^*(\Pi_{\hf{A}(\delta)} \hf{B}(\hf{A}(\id, \delta)))
\end{eqnarray*}
The first isomorphism is by uniqueness of the pullback of $\hf{B}(\id, \delta)$ along $\hf{A}(\delta, \alpha)$, while the second is by the Beck-Chevalley
condition applied to the pullback square of $\hf{A}(\delta, \alpha)$. Let us call $\phi$ this isomorphism. The action of $\alpha^*$
also gives a canonical morphism
$h : \dom(\alpha^*(\Pi_{\hf{A}(\delta)} \hf{B}(\hf{A}(\id, \delta)))) \to \dom(\Pi_{\hf{A}(\delta)} \hf{B}(\hf{A}(\id, \delta)))$, thus we define:
\[
\Pi(\hf{A}, \hf{B})(\delta, \alpha) = h \phi : \dom(\Pi(\hf{A}, \hf{B})(\delta)) \to \dom(\Pi(\hf{A}, \hf{B})(\delta \alpha))
\]
As needed this defines a pullback square since it is obtained as an isomorphism and a pullback, hence the definition of the functorial
family $\Pi(\hf{A}, \hf{B})$ is now complete, since the equations come from the universal property of the pullback. The fact that
$\Pi(\hf{A}, \hf{B})[\delta]$ and $\Pi(\hf{A}[\delta],\hf{B}[\ext{\delta \p, \q}])$ coincide on objects (of $\C/\Gamma$) is a straightforward
calculation, from which the fact that they coincide on morphisms can be directly deduced.

The combinators $\lambda$ and $\ap$ come from natural applications of the adjunction $(\hf{A}(\id))^* \adj \Pi_{\hf{A}(\id)}$, and the computation rules follow
from the properties of adjunctions. Behaviour of the combinators $\lambda$ and $\ap$ under substitution require to rework the proof of the Beck-Chevalley
conditions for lcccs. As in \cite{hofmann:csl}, we will not give the details.

\begin{proposition}
Let $\C$ be a lccc, then $(\C, T_\C)$ supports $\Pi$-types.
\end{proposition}

\subsection{Image of $1$-cells}

\begin{lemma}
Let $(F, \sigma): (\C, T) \to (\D, T')$ be a pseudo cwf-morphism with families of isomorphisms $\theta$ and $\rho$. Then for any $\delta : \Delta \to \Gamma$ in $\C$ and type $A\in \Ty(\Gamma)$,
we have:
\[
F(\subst{\delta \p, \q}) = \rho_{\Gamma, A}^{-1} \subst{F(\delta)\p, \q} \theta_{A, \delta}^{-1} \rho_{\Delta, A[\delta]}
\]
\label{equiv_theta_appendix}
\end{lemma}
\begin{proof}
Direct calculation.
\begin{eqnarray*}
F(\subst{\delta\p, \q}) &=& \rho_{\Gamma, A}^{-1} \subst{F(\delta \p), \{\theta_{A, \delta\p}^{-1}\}(\sigma_{\Delta\cext A[\delta]}^{A[\delta\p]}(\q))}\\
			&=& \rho_{\Gamma, A}^{-1} \subst{F(\delta \p), \{\theta_{A, \delta\p}^{-1}\}(\{\theta_{A[\delta], \p}\}(\q[\rho_{\Delta, A[\delta]}]))}\\
			&=& \rho_{\Gamma, A}^{-1} \subst{F(\delta \p), \{\indexed{T'}(F\p)(\theta_{A, \delta}^{-1})\}(q[\rho_{\Delta, A[\delta]}])}\\
			&=& \rho_{\Gamma, A}^{-1} \subst{F(\delta \p), \q[\indexed{T'}(F\p)(\theta_{A, \delta}^{-1})\subst{id, q[\rho_{\Delta, A[\delta]}]}]}\\
			&=& \rho_{\Gamma, A}^{-1} \subst{F(\delta \p), \q[\subst{\p, \q[\theta_{A, \delta}^{-1}\subst{(F\p)\p, \q}]}\subst{id, q[\rho_{\Delta, A[\delta]}]}]}\\
			&=& \rho_{\Gamma, A}^{-1} \subst{F(\delta \p), \q[\theta_{A, \delta}^{-1}\subst{F\p, \q[\rho_{\Delta, A[\delta]}]}]}\\
			&=& \rho_{\Gamma, A}^{-1} \subst{F(\delta \p), \q[\theta_{A, \delta}^{-1}\rho_{\Delta, A[\delta]}]}\\
			&=& \rho_{\Gamma, A}^{-1} \subst{F(\delta)\p, \q} \theta_{A, \delta}^{-1} \rho_{\Delta, A[\delta]}
\end{eqnarray*}
Using preservation of substitution extension and $\q$, then coherence of $\theta$ and manipulation of cwf combinators.
\end{proof}

From any functor $F: \C \to \D$ preserving finite limits, its extension to $(F, \sigma_F) : (\C, T_\C) \to (\D, T_\D)$ relies heavily on the 
following lemma.

\begin{lemma}[Generation of isomorphisms]
Let $(\C, T)$ and $(\D, T')$ be two cwfs, $F: \C \to \D$ a functor preserving finite limits, a family of functions
$\sigma_\Gamma : \Ty(\Gamma) \to \Ty'(F\Gamma)$ and a family of isomorphisms
$\rho_{\Gamma, A} : F(\Gamma\cext A) \to F\Gamma \cext \sigma_\Gamma(A)$ such that $\p \rho_{\Gamma, A} = F \p$. Then there exists an unique
choice of functions $\sigma_\Gamma^A$ on terms and of isomorphisms $\theta_{A, \delta}$ such that $(F, \sigma)$ is a weak cwf-morphism.
\label{completion_cwfmorphisms_appendix}
\end{lemma}
\begin{proof}
By Lemma \ref{redundant_appendix}, the unique way of extending $\sigma$ to terms is by exploiting the redundancy between terms and sections and
set $\sigma_\Gamma^A(a) = \q[\rho_{\Gamma, A} F(\ext{\id, a})]$. To generate $\theta$, we exploit the two squares below:
\[\hspace{-5pt}
\xymatrix@R=20pt@C=35pt{
F\Delta \cext \sigma_\Gamma({A})[F\delta]       \ar[r]^{\ext{(F \delta) \p\q}}
        \ar[d]_{p}&
F\Gamma\cext \sigma_\Gamma({A}) \ar[d]^{\p}\\
F\Delta \ar[r]_{F\delta}&
F\Gamma
}
\xymatrix@R=20pt@C=80pt{
F\Delta\cext \sigma_\Delta({A}[\delta]) \ar[r]^{\rho_{\Gamma, A} F(\ext{\delta \p, \q}) \rho_{\Delta, A[\delta]}^{-1}}
        \ar[d]_{\p}&
F\Gamma\cext \sigma_\Gamma(A)   \ar[d]^{\p}\\
F\Delta \ar[r]_{F\delta}&
F\Gamma
}
\]
The first square is a standard substitution pullback. The second is a pullback because $F$ preserves finite limits and $\rho_{\Gamma, A}$ and $\rho_{\Delta, A[\delta]}$ are isomorphisms.
The isomorphism $\theta_{A, \delta}$ is then defined as the unique mediating morphism from the first to the second. There is no other possible choice for $\theta_{A, \delta}$ : indeed,
in an arbitrary pseudo cwf-morphism $\theta_{A, \delta}$ necessarily already commutes with the projections of these pullback diagrams (easy consequence of Lemma \ref{equiv_theta_appendix}),
therefore whenever $(F, \sigma)$ is such that $F$ preserves finite limits, $\theta_{A, \delta}$ must coincide with this mediating arrow above.

We must now check that $\theta$, defined as above, satisfies the necessary coherence and naturality conditions. 
Clearly, $\theta$ defined as above commutes with the projections and satisfies $\theta_{A, \id} = \id$. We must now check that $\theta$, defined as above, satisfies the necessary coherence and naturality
conditions, which will be a consequence of the universal property of the pullback above.
For the coherence condition, consider now that we have $\alpha : \Omega \to \Delta$ and $\delta : \Delta \to \Gamma$. The morphism
$\theta_{A[\delta], \alpha} \indexed{T'}{(F\alpha)}(\theta_{A, \delta})$ is an isomorphism between $\Omega \cext \sigma_\Gamma(A)[F(\delta \alpha)]$ and
$\Omega \cext \sigma_\Omega(A[\delta \alpha])$, thus we just have to prove that it preserves the projections (of the pullback corresponding to these
expressions, as in the diagram above) to conclude by uniqueness of such an isomorphism. The two equations to prove are therefore the following.
\begin{eqnarray}
\p \theta_{A[\delta], \alpha} \indexed{T'}(F\alpha)({\theta_{A, \delta}}) &=& \p\label{eqnfst}\\
\rho_{\Gamma, A} F(\ext{\delta\alpha \p, \q}) \rho_{\Omega, A[\delta\alpha]}^{-1} \theta_{A[\delta], \alpha} \indexed{T'}(F\alpha)(\theta_{A, \delta})
 &=& \ext{(F(\delta \alpha)) \p, \q}\label{eqnsnd}
\end{eqnarray}
Equation \eqref{eqnfst} is clear by property of $\theta_{A[\delta], \alpha}$ and construction of $\substopen{(F\alpha)}{\theta_{A, \delta}}$,
while equation \eqref{eqnsnd} is a consequence of Lemma \ref{equiv_theta_appendix}.

It remains to prove that $\theta$ satisfies the naturality condition. Let $f:A \to B$ be a morphism in $\indexed{T}(\Gamma)$. We need to establish the following equality:
\[
\indexed{\sigma}_\Delta(\indexed{T}(\delta)(f)) \theta_{A, \delta} = \theta_{B, \delta} \indexed{T'}(F\delta){(\indexed{\sigma}_\Gamma(f))}
\]
It follows from the fact that both sides of this equation make the two triangles commute in the following diagram, which is a pullback diagram because $F$
preserves finite limits.
\[
\xymatrix@C=80pt{
F\Delta\cext \sigma_\Gamma(A)[F\delta]
        \ar@/^1.3pc/[drr]^{\indexed{\sigma}_\Gamma(f)\ext{(F\delta)\p, \q}}
        \ar@/_/[ddr]_{\p}
        \ar@{.>}[dr]\\
&F\Delta\cext \sigma_\Delta(B[\delta])
        \ar[r]^{\rho_{\Gamma, B} F(\ext{\delta \p, \q}) \rho_{\Delta, B[\delta]}^{-1}}
        \ar[d]^\p&
F\Gamma\cext \sigma_\Gamma(B)
        \ar[d]^\p\\
&F\Delta\ar[r]_{F\delta}&
F\Gamma
}
\]
The proof that the two triangles commute only involves the definition of $\theta_{A, \delta}$ and $\theta_{B, \delta}$, along with
manipulation of cwf combinators. This ends the proof that $(F, \sigma)$ is a weak cwf-morphism.
\end{proof}

\begin{proposition}
If $F: \C \to \D$ preserves finite limits, then $\sigma_F$ preserves democracy.
\end{proposition}
\begin{proof}
The functor $F$ preserves finite limits, thus it preserves in particular the terminal object: let us
denote by $\iota : \nilc \to F\nilc$ the inverse to the terminal projection. Let us note now that since the two involved cwfs have been built with
Hofmann's construction, their democratic structure is trivial; we have $\nilc \cext \bar{\Gamma} = \Gamma$ and $\gamma_\Gamma = \id$.
In particular, we have $F(\nilc\cext \bar{\Gamma}) = F(\Gamma) = \nilc\cext \bar{F \Gamma}$.
Thus to get preservation of the democratic structure, it is natural to choose:
\[
d_\Gamma = \ext{\iota, \q} \rho_{\nilc, \bar{\Gamma}}^{-1} : \nilc\cext \sigma_{\nilc}(\bar{\Gamma}) \to \nilc\cext \bar{F\Gamma}[\subst{}]
\]
which makes the coherence condition essentially trivial.
\end{proof}

\begin{proposition}
If $(F, \sigma): (\C, T) \to (\D, T')$ such that $(\C, T)$ and $(\D, T')$ supports identity types and $F$ preserves finite limits, then $\sigma$ preserves identity types.
\end{proposition}
\begin{proof}
Let $A\in Type(\Gamma)$, and $a, a'\in \Gamma\vdash A$, then $\Gamma\cext \Id_A(a,a')$ along with its projection to $\Gamma$ is an equalizer
of $\ext{\id, a}$ and $\ext{\id, a'}$. Indeed if $\delta : \Delta \to \Gamma$ such that $\ext{\id, a}\delta = \ext{\id, a'}\delta$, it is straightforward to
see that the morphism $h = \ext{\delta, \refl_{A[\delta], a[\delta]}}$ typechecks and satisfies $\p h = \delta$. It is also the unique such morphism
because of the uniqueness of identity proofs. But $F$ is left exact and in particular preserves equalizers, hence the pair
$(F(\Gamma\cext \Id_A(a, a')), F(\p))$ defines an equalizer of $F(\ext{\id, a}) = \rho_{\Gamma, A}^{-1} \ext{\id, \sigma_\Gamma^A(a)}$ and
$F(\ext{\id, a'}) = \rho_{\Gamma, A}^{-1} \ext{\id, \sigma_\Gamma^A(a')}$. From this it is obvious that the pair $(F\Gamma\cext \sigma_\Gamma(A),\p)$
is an equalizer of $\ext{\id, \sigma_\Gamma^A(a)}$ and $\ext{\id, \sigma_\Gamma^A(a')}$. But for the same reason as in the beginning of the proof,
the pair $(F\Gamma\cext \Id_{\sigma_\Gamma(A)}(\sigma_\Gamma^A(a), \sigma_\Gamma^A(a')), \p)$ is already such an equalizer, therefore they must
be isomorphic and $(F, \sigma)$ preserves identity types.
\end{proof}

We will now address the corresponding proposition for preservation of $\Pi$-types by $(F, \sigma)$, provided $F$ preserves lcc structure. The proof will
make use of thr following notion.

\begin{definition}
If $(\C, T)$ is a cwf (not necessarily supporting $\Pi$-types), $\Gamma$ a context in $\C$ and $A\in \Ty(\Gamma)$ and $B\in \Ty(\Gamma\cext A)$, let us call \emph{a $\Pi$-object} of
$A$ and $B$ any type $\Pi(A, B)$ such that for all term $c: \Gamma\cext A \vdash B$ there is $\lambda(c) : \Gamma\vdash \Pi(A, B)$, for all $c: \Gamma\vdash \Pi(A, B)$ and
$a:\Gamma\vdash A$ there is $\ap(c, a):\Gamma\vdash B[\subst{\id, a}]$ satisfying the computation rules for $\Pi$-types (but no requirements w.r.t. substitution). Then
it is straightforward to check that just as exponentials $A\tto B$ of $A$ and $B$ are unique up to isomorphism, $\Pi$-objects of $A$ and $B$ are unique up to type isomorphism.

This notion extends to categories by relating them to the cwf built with Hofmann's construction:
if $\C$ is any category with a terminal object, if $g:A\to B$ and $f: B\to C$ are morphisms in $\C$, we will call a $\Pi$-object of $f$ and $g$ any 
morphism $\Pi(f, g): D \to C$ such that $\hat{\Pi(f, g)}$ is a $\Pi$-object of $\hat{f}$ and $\hat{g}$ in $(\C, T_\C)$.
\end{definition}

\begin{lemma}
The two notions of $\Pi$-objects coincide:
if $(\C, T)$ is a cwf, $\Gamma$ a context in $\C$ and $A\in \Ty(\Gamma)$ and $B\in \Ty(\Gamma\cext A)$, then a type $C$ is a $\Pi$-object of $A$ and $B$
if and only if $\p_C$ is a $\Pi$-object of $\p_A$ and $\p_B$.
\label{pi_objects_appendix}
\end{lemma}
\begin{proof}
Obvious by the correspondence between terms of type $A$ and sections of $\p_A$.
\end{proof}

\begin{lemma}
If $\C$ is any category with a terminal object, if $g:A\to B$ and $f: B\to C$ are morphisms in $\C$, then there is only one $\Pi$-object
of $f$ and $g$ up to isomorphism in $\C/C$.
\label{pi_unicity_appendix}
\end{lemma}
\begin{proof}
The proof exactly mimics the proof of uniqueness of exponential objects.
\end{proof}

\begin{lemma}
If $\C$ and $\C'$ are lcccs and $F: \C \to \C'$ is a functor preserving finite limits, then if $F$ preserves the lcc structure, then it preserves $\Pi$-objects.
\label{pi_obj_pres_appendix}
\end{lemma}
\begin{proof}
Recall that if $\C$ is locally cartesian closed and if $g: A \to B$ and $f: B \to C$ are morphisms in $\C$, then there is a 
morphism $\Pi_{f}(g): - \to \Gamma$, obtained by the following pullback in $\C/C$:
\[
\xymatrix{
\Pi_{f}(g)        	\ar[r]
                        \ar[d]&
(gf)^{f}	   	\ar[d]^{{g}^{f}}\\
1                       \ar[r]^{\Lambda(\id)}&
{f}^{f}
}
\]
this extends to a functor $\Pi_{f}: \C/B \to \C/C$, which is right adjoint to the pullback functor $f^*$.
Exploiting this adjunction, it is straightforward to prove that $\Pi_{f}(f)$ is a $\Pi$-object of
$f$ and $g$ in $\C$. By uniqueness, it is \emph{the} $\Pi$-object of $f$ and $g$, up to isomorphism. But since $F$ preserves
lcc structure it preserves pullbacks and local exponentiation, thus it maps (up to isomorphism) this pullback diagram into the following pullback:
\[
\xymatrix{
F(\Pi_{f}(g))           \ar[r]
                        \ar[d]&
(F(g)F(f))^{F(f)}       \ar[d]^{{F(g)}^{F(f)}}\\
1                       \ar[r]^{\Lambda(\id)}&
{F(f)}^{F(f)}
}
\]
So $F(\Pi_{f}(g))$ is as required a $\Pi$-object of $F(f)$ and $F(g)$.
\end{proof}

\begin{proposition}
If $(F, \sigma): (\C, T) \to (\C', T')$ such that $(\C, T)$ and $(\C', T')$ supports $\Pi$-types and $F$ preserves lccc structure, then $\sigma$ preserves $\Pi$-types.
\end{proposition}
\begin{proof}
Let $\Gamma$ be a context of $\C$, $A\in \Ty(\Gamma)$ and $B\in \Ty(\Gamma\cext A)$, obviously $\Pi(A, B)$ is a $\Pi$-object of $A$ and $B$. Hence, 
$\p_{\Pi(A, B)}$ is a $\Pi$-object of $\p_A$ and $\p_B$ by Lemma \ref{pi_objects_appendix}. But $F$ preserves $\Pi$-objects by Lemma \ref{pi_obj_pres_appendix}, so
$F(\p_{\Pi(A, B)})$ is a $\Pi$-object of $F(\p_A)$ and $F(\p_B)$. But $F(\p_A)$ is isomorphic to $\p_{\sigma_\Gamma(A)}$ (the isomorphism
being $\rho_{\Gamma, A}$) and $F(\p_B)$ is isomorphic to $\p_{\sigma_{\Gamma\cext A}(B)[\rho_{\Gamma, A}]}$ (the isomorphism 
being $\subst{\rho_{\Gamma, A}^{-1}\p, q}\rho_{\Gamma\cext A, B}^{-1}$), therefore $F(\p_{\Pi(A, B)})$ is a $\Pi$-object of $\p_{\sigma_\Gamma(A)}$ and
$\p_{\sigma_{\Gamma\cext A}(B)[\rho_{\Gamma, A}]}$, hence it must be isomorphic to $\p_{\Pi(\sigma_\Gamma(A), \sigma_{\Gamma\cext A}(B)[\rho_{\Gamma, A}])}$
by Lemma \ref{pi_unicity_appendix}, so we have the required isomorphism $F(\Gamma\cext \Pi(A, B)) \to F\Gamma\cext \Pi(\sigma_\Gamma(A), \sigma_{\Gamma\cext A}(B)[\rho_{\Gamma, A}])$.
\end{proof}

\subsection{Image of $2$-cells}

\begin{lemma}[Completion of pseudo cwf-transformations]
Suppose $(F, \sigma)$ and $(G, \tau)$ are pseudo cwf-morphisms from  $(\C, T)$ to $(\C', T)$  such that $F$ and $G$ preserve finite limits and $\phi: F\natto G$
is a natural transformation, then there exists a family of morphisms $(\psi_\phi)_{\Gamma, A} : \sigma_\Gamma(A) \to \tau_\Gamma(A)[\phi_\Gamma]$
such that $(\phi, \psi_\phi)$ is a pseudo cwf-transformation from $(F, \sigma)$ to $(G, \tau)$.
%
\label{completion_transformations_appendix}
\end{lemma}
\begin{proof}
We set $\psi_{\Gamma, A} = \subst{\p, \q[\rho'_{\Gamma, A}\phi_{\Gamma\cext A}\rho_{\Gamma, A}^{-1}]}: F\Gamma\cext \sigma_\Gamma A\to F\Gamma\cext \tau_\Gamma(A)[\phi_\Gamma]$.
To check the coherence law, consider the following composition of pullback squares.
\[
\xymatrix@C=40pt{
F\Delta\cext \tau_\Delta(A[\delta])[\phi_\Delta]
       \ar[r]^{\subst{\phi_\Delta \p, \q}}
       \ar[d]^{\p}&
G\Delta\cext \tau_\Delta(A[\delta])
       \ar[rr]^{\rho_{\Gamma, A}' G(\subst{\delta \p, \q}) (\rho_{\Delta, A[\delta]}')^{-1}}
       \ar[d]^\p&&
G\Gamma\cext \tau_\Gamma(A)
       \ar[d]^\p\\
F\Delta        \ar[r]_{\phi_\Delta}&
G\Delta        \ar[rr]_{G\delta}&&
G\Gamma
}
\]
The two paths $\indexed{T'}(\phi_\Delta)(\theta'_{A, \delta})\indexed{T'}(F\delta)(\psi_{\Gamma, A})$ and $\psi_{\Delta, A[\delta]}\theta_{A, \delta}$ of the coherence
diagram behave in the same way with respect to this pullback. Here is the calculation for the first path of the coherence diagram:
\begin{eqnarray*}
&&\rho_{\Gamma, A}' G(\subst{\delta \p, \q}) (\rho_{\Delta, A[\delta]}')^{-1} \subst{\phi_\Delta \p, \q} \indexed{T'}(\phi_\Delta)(\theta'_{A, \delta})\indexed{T'}(F\delta)(\psi_{\Gamma, A})\\
&=& \subst{(G\delta)\p, \q} {\theta'_{A, \delta}}^{-1} \subst{\phi_\Delta \p, \q} \indexed{T'}(\phi_\Delta)(\theta'_{A, \delta})\indexed{T'}(F\delta)(\psi_{\Gamma, A})\\
&=& \subst{(G\delta)\p, \q} {\theta'_{A, \delta}}^{-1} \subst{\phi_\Delta \p, \q} \subst{\p, \q[\theta'_{A, \delta}\subst{\phi_\Delta\p, \q}} \indexed{T'}(F\delta)(\psi_{\Gamma, A})\\
&=& \subst{(G\delta)\p, \q} {\theta'_{A, \delta}}^{-1} \subst{\phi_\Delta p, q[\theta'_{A, \delta}\subst{\phi_\Delta\p, \q}} \indexed{T'}(F\delta)(\psi_{\Gamma, A})\\
&=& \subst{(G\delta)\p, \q} {\theta'_{A, \delta}}^{-1} \subst{\p\theta'_{A, \delta}\subst{\phi_\Delta\p, \q}, q[\theta'_{A, \delta}\subst{\phi_\Delta\p, \q]}} \indexed{T'}(F\delta)(\psi_{\Gamma, A})\\
&=& \subst{(G\delta)\p, \q} \subst{\phi_\Delta\p, \q}\indexed{T'}(F\delta)(\psi_{\Gamma, A}) \\
&=& \subst{(G\delta)\p, \q} \subst{\phi_\Delta\p, \q} \subst{\p, \q[\psi_{\Gamma, A}\subst{(F\delta)\p, \q}]}\\
&=& \subst{(G\delta)\phi_\Delta\p, \q[\rho_{\Gamma, A}'\phi_{\Gamma\cext A}\rho_{\Gamma, A}^{-1}\subst{(F\delta)\p, \q}]}\\
&=& \subst{\phi_\Gamma (F\delta) \p, \q[\rho_{\Gamma, A}'\phi_{\Gamma\cext A}\rho_{\Gamma, A}^{-1}\subst{(F\delta)\p, \q}]}\\
&=& \subst{\phi_\Gamma \p, \q[\rho_{\Gamma, A}'\phi_{\Gamma\cext A}\rho_{\Gamma, A}^{-1}]}\subst{(F\delta)\p, \q}\\
&=& \rho_{\Gamma, A}' \phi_{\Gamma\cext A} \rho_{\Gamma, A}^{-1} \subst{(F\delta)\p, \q}
\end{eqnarray*}
where we use Lemma \ref{equiv_theta_appendix}, then only the definition of $\psi_{\Gamma, A}$, naturality of $\phi$ and manipulation of cwf combinators. The calculation for the
other path follows:
\begin{eqnarray*}
&&\rho_{\Gamma, A}' G(\subst{\delta \p, \q}) {\rho_{\Delta, A[\delta]}'}^{-1} \subst{\phi_\Delta \p, \q} \psi_{\Delta, A[\delta]}\theta_{A, \delta}\\
&=& \rho_{\Gamma, A}' G(\subst{\delta \p, \q}) {\rho_{\Delta, A[\delta]}'}^{-1} \subst{\phi_\Delta \p, \q} \subst{\p, \q[\rho_{\Delta, A[\delta]}'\phi_{\Delta\cext A[\delta]} \rho_{\Delta, A[\delta]}^{-1}]}\theta_{A, \delta}\\
&=& \rho_{\Gamma, A}' G(\subst{\delta \p, \q}) {\rho_{\Delta, A[\delta]}'}^{-1} \subst{\phi_\Delta \p, \q[\rho_{\Delta, A[\delta]}'\phi_{\Delta\cext A[\delta]} \rho_{\Delta, A[\delta]}^{-1}]} \theta_{A, \delta}\\
&=& \rho_{\Gamma, A}' G(\subst{\delta \p, \q}) {\rho_{\Delta, A[\delta]}'}^{-1} \subst{\p \rho_{\Delta, A[\delta]}'\phi_{\Delta\cext A[\delta]} \rho_{\Delta, A[\delta]}^{-1}, \q[\rho_{\Delta, A[\delta]}'\phi_{\Delta\cext A[\delta]} \rho_{\Delta, A[\delta]}^{-1}]} \theta_{A, \delta}\\
&=& \rho_{\Gamma, A}' G(\subst{\delta \p, \q}) {\rho_{\Delta, A[\delta]}'}^{-1} \rho_{\Delta, A[\delta]}'\phi_{\Delta\cext A[\delta]} \rho_{\Delta, A[\delta]}^{-1}\theta_{A, \delta}\\
&=& \rho_{\Gamma, A}' G(\subst{\delta \p, \q}) \phi_{\Delta\cext A[\delta]} \rho_{\Delta, A[\delta]}^{-1}\theta_{A, \delta}\\
&=& \rho_{\Gamma, A}' \phi_{\Gamma\cext A} F(\subst{\delta \p, \q}) \rho_{\Delta, A[\delta]}^{-1}\theta_{A, \delta}\\
&=& \rho_{\Gamma, A}' \phi_{\Gamma\cext A} \rho_{\Gamma, A}^{-1} \rho_{\Gamma, A} F(\subst{\delta \p, \q}) \rho_{\Delta, A[\delta]}^{-1}\theta_{A, \delta}\\
&=& \rho_{\Gamma, A}' \phi_{\Gamma\cext A} \rho_{\Gamma, A}^{-1} \subst{(F\delta)\p, \q}
\end{eqnarray*}
We have used naturality of $\phi$, preservation of the first projection by $(F, \sigma)$ and $(G, \tau)$ and manipulations on cwf combinators.
\end{proof}

\begin{lemma}
Completion of pseudo cwf-transformations commutes with both notions of composition, \emph{i.e.} if $\phi: F \natto G$ and $\phi' : G\natto H$, then
\begin{eqnarray*}
(\phi', \psi_{\phi'})(\phi, \psi_\phi) &=& (\phi'\phi, \psi_{\phi'\phi})\\
(\phi, \psi_\phi)\star 1 &=& (\phi\star 1, \psi_{\phi\star 1})\\
1\star (\phi, \psi_\phi) &=& (1\star \phi, \psi_{1\star \phi})\\
(\phi', \psi_{\phi'})\star(\phi, \psi_\phi)&=& (\phi'\star\phi, \psi_{\phi'\star\phi})
\end{eqnarray*}
whenever these expressions typecheck.
\label{comm_completion_appendix}
\end{lemma}
\begin{proof}
The first equality is just a straightforward verification, and the second and third are trivial given the definition of $\psi_1$. The fourth though,
requires a more involved calculation with arguments really similar to those used to prove Lemma \ref{completion_transformations_appendix}.
We only detail the third case. Imagine we have the following situation:
\[
\xymatrix{
(\C, T)	\ar@/^/[r]^{(F, \sigma)}
	\ar@/_/[r]_{(G, \tau)}&
(\C', T')	\ar@/^/[r]^{(F', \sigma')}
	        \ar@/_/[r]_{(G', \tau')}&
(\C'', T'')
}
\]
Let us call $\theta$ and $\rho$ the components of $(F, \sigma)$, $\theta'$ and $\rho'$ the components of $(F', \sigma')$, $t$ and $r$ the components
of $(G, \tau)$ and $t'$ and $r'$ the components of $(G', \tau')$. Let us also consider natural transformations $\phi: F \natto G$ and $\phi': F'\natto G'$.
Let us recall that the vertical composition of pseudo cwf-transformations follow those of $2$-cells in the $2$-category of indexed categories over
arbitrary bases, which means $(\phi, \psi_{\phi})\star(\phi', \psi_{\phi'}) = (\phi\star \phi', m)$, where $m_{\Gamma, A}$ is obtained by:
\[
\xymatrix{
\sigma_{F\Gamma}'(\sigma_\Gamma A)
	\ar[r]^{\indexed{\sigma'}_{F\Gamma}((\psi_\phi)_{\Gamma, A})}&
\sigma_{F\Gamma}'(\tau_\Gamma A[\phi_\Gamma])
	\ar[r]^{{\theta_{\tau_\Gamma A, \phi_\Gamma}'}^{-1}}&
\sigma_{G\Gamma}'(\tau_\Gamma A)[F'\phi_\Gamma]
	\ar[r]^{\indexed{T''}(F'\phi_\Gamma)((\psi_{\phi'})_{G\Gamma, \tau_\Gamma A})}&
\tau_{G\Gamma}'(\tau_\Gamma A)[\phi_{G\Gamma}'F'(\phi_\Gamma)]
}
\]
which the following calculation relates to $(\psi_{\phi\star \phi'})_{\Gamma, A}$:
\begin{eqnarray*}
m_{\Gamma, A} &=& \indexed{T''}(F'\phi_\Gamma)((\psi_{\phi'})_{G\Gamma, \tau_\Gamma A}) {\theta_{\tau_\Gamma A, \phi_\Gamma}'}^{-1} \indexed{\sigma'}_{F\Gamma}((\psi_\phi)_{\Gamma, A})\\
&=& \subst{\p, \q[(\psi_{\phi'})_{G\Gamma, \tau_\Gamma A}\subst{(F'\phi_\Gamma)\p, \q}]} {\theta_{\tau_\Gamma A, \phi_\Gamma}'}^{-1} \rho'_{F\Gamma, \tau_\Gamma A[\phi_\Gamma]} F'((\psi_\phi)_{\Gamma, A}) {\rho'_{F\Gamma, \sigma_\Gamma A}}^{-1}\\
&=& \subst{\p, \q[(\psi_{\phi'})_{G\Gamma, \tau_\Gamma A}\subst{(F'\phi_\Gamma)\p, \q}{\theta_{\tau_\Gamma A, \phi_\Gamma}'}^{-1} \rho'_{F\Gamma, \tau_\Gamma A[\phi_\Gamma]} F'((\psi_\phi)_{\Gamma, A}) {\rho'_{F\Gamma, \sigma_\Gamma A}}^{-1}]}\\
&=& \subst{\p, \q[(\psi_{\phi'})_{G\Gamma, \tau_\Gamma A}\rho_{G\Gamma, \tau_\Gamma A}' F'(\subst{\phi_\Gamma \p, \q}) {\rho'_{F\Gamma, \tau_\Gamma A[\phi_\Gamma]}}^{-1} \rho'_{F\Gamma, \tau_\Gamma A[\phi_\Gamma]} F'((\psi_\phi)_{\Gamma, A}) {\rho'_{F\Gamma, \sigma_\Gamma A}}^{-1}]}\\
&=& \subst{\p, \q[(\psi_{\phi'})_{G\Gamma, \tau_\Gamma A}\rho_{G\Gamma, \tau_\Gamma A}' F'(\subst{\phi_\Gamma \p, \q})F'((\psi_\phi)_{\Gamma, A}) {\rho'_{F\Gamma, \sigma_\Gamma A}}^{-1}]}\\
&=& \subst{\p, \q[\subst{\p, \q[r'_{G\Gamma, \tau_\Gamma A} \phi'_{G\Gamma\cext  \tau_\Gamma A} {\rho'_{G\Gamma, \tau_\Gamma A}}^{-1}\rho_{G\Gamma, \tau_\Gamma A}' F'(\subst{\phi_\Gamma \p, \q})F'((\psi_\phi)_{\Gamma, A}) {\rho'_{F\Gamma, \sigma_\Gamma A}}^{-1}]}]}\\
&=& \subst{\p, \q[r'_{G\Gamma, \tau_\Gamma A} \phi'_{G\Gamma\cext  \tau_\Gamma A}F'(\subst{\phi_\Gamma \p, \q})F'((\psi_\phi)_{\Gamma, A}) {\rho'_{F\Gamma, \sigma_\Gamma A}}^{-1}]}\\
&=& \subst{\p, \q[r'_{G\Gamma, \tau_\Gamma A} \phi'_{G\Gamma\cext  \tau_\Gamma A}F'(\subst{\phi_\Gamma \p, \q} r_{\Gamma, A} \phi_{\Gamma\cext A} \rho_{\Gamma, A}^{-1}) {\rho'_{F\Gamma, \sigma_\Gamma A}}^{-1}]}\\
&=& \subst{\p, \q[r'_{G\Gamma, \tau_\Gamma A} \phi'_{G\Gamma\cext  \tau_\Gamma A}F'(r_{\Gamma, A} \phi_{\Gamma\cext A} \rho_{\Gamma, A}^{-1}){\rho'_{F\Gamma, \sigma_\Gamma A}}^{-1}]}\\
&=& \subst{\p, \q[r'_{G\Gamma, \tau_\Gamma A} G'(r_{\Gamma, A}) \phi'_{G(\Gamma\cext A)} F'(\phi_{\Gamma\cext A}) F'(\rho_{\Gamma, A}^{-1}) {\rho'_{F\Gamma, \sigma_\Gamma A}}^{-1}]}\\
&=& \subst{\p, \q[\rho^{G'G}_{\Gamma, A} (\phi \star \phi')_{\Gamma\cext A} {\rho^{F'F}_{\Gamma, A}}^{-1}]}\\
&=& (\psi_{\phi\star \phi'})_{\Gamma, A}
\end{eqnarray*}
We have first unfolded the action of $\indexed{T''}$ and $\indexed{\sigma'}$, then applied Lemma \ref{equiv_theta_appendix}, unfolded the definition of
${\psi_{\phi'}}$ and $\psi_{\phi}$, then used naturality of $\phi'$ and $\phi$. Of course there are a lot of simplification steps, involving
the preservation of the first projection by all the present pseudo cwf-morphisms and manipulation of cwf combinators.
\end{proof}

\begin{proposition}
There are pseudofunctors $H: \FL \to \CWFFL$ and $H: \LCC \to \CWFLCC$ defined by:
  \begin{eqnarray*}
    H \C &=& (\C, T_\C)\\
    H F &=& (F, \sigma_F)\\
    H \phi &=& (\phi, \psi_\phi)
  \end{eqnarray*}
\end{proposition}
\begin{proof}
First, note that as proved in Lemma \ref{comm_completion_appendix}, $H$ is functorial on $2$-cells. 

For each $\C$ we need an invertible $2$-cell $H_\C : Id_{(\C, T_\C)} \to H(Id_\C)$, this will be the identity $2$-cell since we have in fact
$H(Id_\C) = (Id_\C, \sigma_{Id_\C}) = id_{\C, T_\C}$ by construction of $\sigma_{Id_\C}$.

For each two functors $F: \C\to \D$ and $G: \D\to \E$ we need an isomorphism $H_{F, G} : HG\circ HF \to H(G\circ F)$, natural in $F$ and $G$.
It is given by $H_{F, G} = (1_{GF}, \psi_{1_{GF}})$. The naturality condition amounts to the fact that the following square commutes:
\[
\xymatrix@C=40pt{
(G, \sigma_G)(F, \sigma_F)
	\ar[r]^{(1_{GF}, \psi_{1_{GF}})}
	\ar[d]^{(\phi, \psi_\phi)\star(\phi', \psi_{\phi'})}&
(GF, \sigma_{GF})
	\ar[d]^{(\phi'\star\phi, \psi_{\phi'\star\phi})}\\
(G', \sigma_{G'})(F', \sigma_{F'})
	\ar[r]^{(1_{GF}, \psi_{1_{GF}})}&
(G'F', \sigma_{G'F'})
}
\]
which is a direct consequence of Lemma \ref{comm_completion_appendix}. The coherence laws w.r.t. associativity of composition and identities also
stems from Lemma \ref{comm_completion_appendix}. In fact, Lemma \ref{comm_completion_appendix} implies that to check the validity of any equation involving vertical and
horizontal compositions of pseudo cwf-transformations built with Lemma \ref{completion_transformations_appendix} and identity pseudo cwf-transformations,
it suffices to check the equality of the corresponding base natural transformation, ignoring the modifications.
\end{proof}

\section{Proofs of Section 6}

\begin{definition}[The pseudo cwf-morphism $\eta_{(\C, T)}$]
For each cwf $(\C, T)$, context $\Gamma$ of $\C$ and type $A\in \Ty(\Gamma)$. Consider:
\begin{itemize}
\item The identity functor $Id_\C : \C \to \C$,
\item For each context $\Gamma$ and type $A\in \Ty(\Gamma)$, the functorial family $\sigma_\Gamma(A)$ defined by:
\begin{eqnarray*}
\sigma_\Gamma(A)(\delta) &=& \p_{A[\delta]}\\
\sigma_\Gamma(A)(\delta, \gamma) &=& \subst{\gamma\p, \q}
\end{eqnarray*}
\item The isomorphism $\rho_{\Gamma, A} = \id_{\Gamma\cext A}$.
\end{itemize}
Then by Lemma \ref{completion_cwfmorphisms_appendix}, it completes in an unique way to a pseudo cwf-morphism $\eta_{(\C, T)} : (\C, T) \to (\C, T_\C) = HU((\C, T))$.
\end{definition}

\begin{lemma}[The pseudonatural transformation $\eta$]
The family $\eta_{(\C, T)} : (\C, T) \to HU((\C, T))$ is pseudonatural in $(\C, T)$.
\end{lemma}
\begin{proof}
For each pseudo cwf-morphism $(F, \sigma)$, the pseudonaturality square relates two pseudo cwf-morphisms whose base functor
is $F$. Hence, the necessary invertible pseudo cwf-transformation is obtained using Lemma \ref{completion_transformations_appendix} from
the identity natural transformation on $F$. The coherence conditions are straightforward consequences of Lemma \ref{comm_completion_appendix}.
\end{proof}

\begin{definition}[The pseudo cwf-morphism $\epsilon_{(\C, T)}$]
For each cwf $(\C, T)$, context $\Gamma$ of $\C$ and type $A\in \Ty(\Gamma)$. Consider:
\begin{itemize}
\item The identity functor $Id_\C : \C \to \C$,
\item For each context $\Gamma$ and functorial family $\hf{A}:\C/\Gamma\to \arrow{\C}$, the type $\tau_\Gamma(\hf{A})$ defined by:
\[
\tau_\Gamma(\hf{A}) = \Inverse{(\hf{A}(\id_\Gamma))}
\]
\item The isomorphism $\rho_{\Gamma, A}: \dom(\hf{A}) \to \Gamma\cext \Inverse{(\hf{A}(\id_\Gamma))}$ is the isomorphism between $\hf{A}(\id)$ and 
$\p_{(\Inverse{\hf{A}(\id_\Gamma)})}$ in $\C/\Gamma$.
\end{itemize}
Then by Lemma \ref{completion_cwfmorphisms_appendix}, it completes in an unique way to a pseudo cwf-morphism $\epsilon_{(\C, T)} : HU(\C, T) = (\C, T_\C) \to (\C, T)$.
\end{definition}

\begin{lemma}[The pseudonatural transformation $\epsilon$]
The family $\epsilon_{(\C, T)} : HU(\C, T) \to (\C, T)$ is pseudonatural in $(\C, T)$.
\end{lemma}
\begin{proof}
Exactly the same reasoning as for $\eta$.
\end{proof}

\begin{theorem}
We have the following biequivalences of $2$-categories.
\[
\begin{array}{lcr}
\xymatrix{
\FL\ 
\ar@<0.5ex>[r]^{H\ \ \ \ }&
\ \CWFFL
\ar@<0.5ex>[l]^{U\ \ \ \ }
}&~~~~~~~~&
\xymatrix{
\LCC\ 
\ar@<0.5ex>[r]^{H\ \ \ \ }&
\ \CWFLCC
\ar@<0.5ex>[l]^{U\ \ \ \ }
}
\end{array}
\]
\end{theorem}
\begin{proof}
We need to define invertible modifications $m: \eta\epsilon \to 1$ and $m':\epsilon\eta \to 1$. Taken at each cwf $(\C, T)$, 
$m_{(\C, T)}$ and $m'_{(\C, T)}$ are both generated from the identity natural transformation using Lemma \ref{completion_transformations_appendix}.
It is obvious by Lemma \ref{comm_completion_appendix} that they satisfy the required coherence law.
\end{proof}